\documentclass{IEEEtran}

\usepackage{times}
\usepackage{latexsym}
\usepackage{xspace}
\usepackage{booktabs}
\usepackage{amsmath,amssymb}
\usepackage{amsbsy}
\usepackage[algo2e,boxed]{algorithm2e}
\usepackage[boxed]{algorithm}
\usepackage{algorithmic}
\usepackage{graphicx}
\usepackage{multirow}
\usepackage[tight]{subfigure}
\usepackage{color}
\usepackage{url}
\usepackage{balance}

\renewcommand{\paragraph}[1]{\vspace{1mm}\noindent{\bf #1}}
\newcommand{\eat}[1]{}

\newcommand{\A}{{\cal A}}
\newcommand{\C}{C}

\newcommand{\h}{h}
\newcommand{\M}{{\cal M}}
\newcommand{\Md}{m}
\newcommand{\export}{{\sc export}\xspace}

\newcommand{\Stop}{{Stop}}

\newcommand{\lr}{\mbox{{\bf l}}}
\newcommand{\nbrs}{nbrs}
\newcommand{\set}{GT}
\newcommand{\TC}{C}

 \renewcommand{\baselinestretch}{0.96}

\newcommand{\nnodes}{\ensuremath{|V|}}
\newcommand{\nedges}{\ensuremath{|E|}}

\newcommand{\vr}[1]{{#1}}

\newcommand{\Hma}{Hash-\-Min\xspace}
\newcommand{\Ha}{{{Hash-\-to-\-All}}\xspace}
\newcommand{\Hm}{Hash-\-to-\-Min\xspace}
\newcommand{\Hcm}{{{Hash-\-Greater-\-to-\-Min}}\xspace}

\newcommand{\ha}{{\cal H}_{all}\xspace}
\newcommand{\hm}{{\cal H}_{min}\xspace}

\newcommand{\core}[1]{\TC_{\ge #1}}

\newcommand{\red}[1]{\textcolor{red}{#1}}

\newcommand{\ashwinnote}[1]{\red{Ashwin: #1}}

\newtheorem{definition}{Definition}[section]
\newtheorem{proposition}[definition]{Proposition}
\newtheorem{lemma}[definition]{Lemma}

\newtheorem{theorem}[definition]{Theorem}

\newtheorem{example}[definition]{Example}

\newcommand{\squishlist}{
   \begin{list}{$\bullet$}
    { \setlength{\itemsep}{-1pt}      \setlength{\parsep}{3pt}
      \setlength{\topsep}{1pt}       \setlength{\partopsep}{0pt}
      \setlength{\leftmargin}{1.5em} \setlength{\labelwidth}{1em}
      \setlength{\labelsep}{0.5em} } }

\newcommand{\squishend}{
    \end{list}  }

\title{Finding Connected Components in Map-Reduce in Logarithmic Rounds}

\author{
Vibhor Rastogi \hspace{5mm} Ashwin Machanavajjhala \hspace{5mm} Laukik Chitnis
\hspace{5mm}Anish Das Sarma \\  
       {\{vibhor.rastogi, ashwin.machanavajjhala, laukik, anish.dassarma\}@gmail.com}
}

\begin{document}
\maketitle

\begin{abstract}
Given a large graph $G = (V,E)$ with millions of nodes and edges, how do we compute its connected components efficiently?  Recent work addresses this problem in map-reduce, where a fundamental trade-off exists between the number of  map-reduce rounds and the communication of each round. Denoting $d$ the diameter of the graph, and $n$ the number of nodes in the largest component, all prior techniques for map-reduce either require a linear, $\Theta(d)$, number of rounds, or a quadratic, $\Theta(n\nnodes + \nedges)$, communication per round.

We propose here two efficient map-reduce algorithms: {\em(i)} \Hcm, which is a randomized algorithm based on PRAM techniques, requiring
$O(\log{n})$ rounds and $O(\nnodes+\nedges)$ communication per round, and
{\em (ii)} \Hm, which is a novel algorithm, provably finishing in $O(\log{n})$ iterations for path graphs. The proof
technique used for \Hm is novel, but not tight, and it is actually faster than \Hcm in practice.
We conjecture that it requires $2 \log d$ rounds and $3(\nnodes+\nedges)$ communication per round, as demonstrated
in our experiments. Using secondary sorting, a standard map-reduce feature, we scale \Hm to graphs with very large connected
components. 

Our techniques for connected components can be applied to clustering as well. We propose a novel algorithm for agglomerative
single linkage clustering in map-reduce. This is the first map-reduce algorithm for clustering in at most $O(\log n)$ rounds,
where $n$ is the size of the largest cluster. We show the effectiveness of all our algorithms through detailed experiments on
large synthetic as well as real-world datasets.

\eat{Given a large graph $G = (V,E)$ with millions of nodes and edges, how do we compute its connected components efficiently?  Recent work addresses this problem in map-reduce, where a fundamental trade-off exists between the number of  map-reduce rounds and the communication of each round. Denoting $d$ the diameter of the graph, and $n$ the number of nodes in the largest component, all prior techniques either require $d$ rounds, or require $\Theta(n\nnodes + \nedges)$ communication per round. We propose two randomized map-reduce algorithms -- {\em (i)} \Hcm, which provably requires at most $3\log{n}$ rounds  with high probability, and at most $2(\nnodes + \nedges)$ communication per round, and {\em (ii)} \Hm, which has a worse theoretical complexity, but  in practice completes in at most $2 \log d$ rounds and $3(\nnodes + \nedges)$ communication per rounds.

Our techniques for connected components can be applied to clustering as well. We propose a novel algorithm for agglomerative single linkage clustering in map-reduce.   While there have been PRAM-based algorithms proposed in the past for connected components, we are the first to propose a direct Map-Reduce algorithm for clustering in at most $O(\log n)$ rounds, where $n$ is the size of the largest cluster. We show the effectiveness of all our algorithms through detailed experiments on large synthetic as well as real-world datasets.
}
\end{abstract}

\section{Introduction}
\label{sec:intro}

Given a large graph $G = (V,E)$ with millions of nodes and edges, how do we compute its connected components efficiently?
With the proliferation of large databases of linked data, it has become very important to scale to large graphs.
Examples of large datasets include the graph of webpages, where edges are hyperlinks between documents, social networks that
link entities like people, and Linked Open Data\footnote{http://linkeddata.org/} that represents a collection of linked
structured entities.
The problem of finding connected components on such graphs, and the related
problem of {\em undirected $s$-$t$ connectivity}~(USTCON~\cite{lewisP82:symmetric}) that checks whether
two nodes $s$ and $t$ are connected, are fundamental as they are basic building
blocks for more complex graph analyses, like clustering.

The number of vertices $\nnodes$ and edges $\nedges$ in these graphs is very large. Moreover, such graphs often arise as intermediate outputs of large batch-processing tasks (e.g., clustering Web pages and entity resolution), thus requiring us to design algorithms in a distributed setting. Map-reduce \cite{Dean2008:MR} has become a very popular choice for distributed data processing.
In map-reduce, there are two critical metrics to be optimized -- number of map-reduce rounds, since each additional job
incurs significant running time overhead because of synchronization and congestion issues, and communication per round,
since this determines the size of the intermediate data.

There has been prior work on finding connected components iteratively
in map-reduce, and a fundamental trade-off exists between the number
of rounds and the communication per round.  Starting from small clusters, these techniques iteratively expand existing clusters,
by either adding adjacent one-hop graph neighbors, or by merging existing overlapping
clusters. The former kind~\cite{twiddle,pegasus,mpi} require $\Theta(d)$ map-reduce rounds for a graph with diameter $d$, while the latter~\cite{foto11:recursive} require a larger,  $\Theta(n\nnodes+\nedges)$, computation per round, with $n$ being the number of nodes in the largest component.

More efficient $O(\log{n})$ time PRAM algorithms have been proposed for computing connected components. 
While theoretical results simulating $O(\log{n})$ PRAM algorithms in map-reduce using $O(\log{n})$ rounds exist~\cite{karloffSV10:mapreduce}, the PRAM algorithms for connected components have not yet been ported to a
practical and efficient map-reduce algorithm. (See Sec.~\ref{sec:PRAM} for a
more detailed description)


In this paper, we present two new map-reduce algorithms for computing connected
components. The first algorithm, called \Hcm, is an efficient map-reduce
implementation of existing PRAM algorithms~\cite{kargerNP99:cc,Shiloach82}, and provably requires at most $3\log{n}$ rounds with high probability, and a per round communication cost\footnote{Measured as total number of $c$-bit messages communicated, where $c$ is number of bits to represent a single
node in the graph.} of at most $2(\nnodes + \nedges)$.  The second algorithm, called \Hm, is novel, and provably
finishes in $O(\log{n})$ rounds for path graphs. The proof technique used
for \Hm is novel, but not tight, and our experiments show that it requires at most $2 \log
d$ rounds and $3(\nnodes + \nedges)$ communication per rounds.

Both of our map-reduce algorithms iteratively merge overlapping clusters
to compute connected components. Low communication cost is achieved by ensuring
that a single cluster is replicated exactly once, using tricks like
pointer-doubling, commonly used in the PRAM literature. The more intricate problem is
processing graphs for which connected components are so big that either
(i) they do not fit in the memory of a single machine, and hence cause failures, or (ii) they result
in heavy data skew with some clusters being small, while others being large.

The above problems mean that we need to merge overlapping clusters, i.e. remove duplicate nodes occurring in multiple clusters, without materializing entire clusters in memory.
Using \Hm, we solve this problem by maintaining each cluster as key-value pairs, where the key is a common cluster id and values are nodes. Moreover, the values are kept sorted (in lexicographic order), using a map-reduce capability called secondary-sorting, which incurs no extra computation cost. Intuitively, when clusters are merged by the algorithm, mappers individually get values~(i.e, nodes) for a key, sort them, and send them to the reducer for that key. Then the reducer  gets the `merge-sorted' list of all values, corresponding to nodes from all clusters needing to be merged. In a single scan, the reducer then removes any duplicates from the sorted list, without materializing entire clusters in memory.


\eat{
Since, in the worst-case, diameter $d$ can be as large as $n$ for a path
graph, a reduction of rounds from $d$ to $3\log{n}$ (or $2 \log d$) is a significant improvement.
While diameters are not as large as $n$ for real-world large-scale graphs,  in our experiments we observe
a significant improvement for the runtime compared with existing algorithms due to reduction in number of rounds.
}

We also present two novel map-reduce algorithms for single-linkage agglomerative
clustering using similar ideas. One using \Ha that provably completes in $O(\log
n)$ map-reduce rounds, and at most $O(n\nnodes + \nedges)$ communication per round, and other using \Hm that we conjecture completes in $O(\log{d})$ map-reduce rounds, and at
most $O(\nnodes + \nedges)$ communication per round. We believe that these are the first Map-Reduce algorithm for
single linkage clustering that finish in $o(n)$ rounds.


All our algorithms can be easily adapted to the Bulk Synchronous Parallel paradigm
used by recent distributed graph processing systems like Pregel \cite{Malewicz2010:Pregel} and
Giraph \cite{Ching2010:Giraph}. We choose to focus on the
map-reduce setting, \vr{since it is more impacted by a reduction in number of iterations, thus more readily showing the
gains brought by our algorithms}  (see Sections~\ref{sec:relatedwork} and \ref{sec:bspexp} for a more detailed discussion).

\eat{Many prior techniques can be shown to be instantiations of this framework for specific hashing and aggregation functions. Techniques like~\cite{twiddle,pegasus,mpi} can be thought of sending small messages using the hashing function to expand connected components by the one-hop neighbors of nodes in the connected components. Thus they have low complexity per round of $O(\nnodes + \nedges)$, as small messages are used, but require exactly $d$ map-reduce rounds, as expansions occurs by one-hop neighbors. An alternative strategy is to send entire connected components through the hashing function, and expand them by merging one or more of them using the aggregation function. This has a higher communication complexity of $O(n\nnodes + \nedges)$), as larger messages are used, but takes only $\log{d}$ number of rounds, as expansion occur by merging with other connected components.

Our algorithms for connected components (\Hm and \Hcm) and single linkage clustering combine the hashing strategies mentioned above by sending the clusters to only one reducer, but also sending some suitably chosen representatives to other clusters. This strategy gives us best of both worlds, i.e., resulting in low communication complexity along with a small number of map-reduce rounds.

\begin{table}[t]
\centering
\begin{tabular}{ |c|c|c| }
\hline
Symbol & Semantics \\
\hline
$\nnodes$ & number of vertices\\
$\nedges$ & number of edges \\
$n$ & size of the largest connected component \\
$d$ & diameter of the graph\\
\hline
\end{tabular}
\caption{Notation}
\end{table}
}

\paragraph{Contributions and Outline:}

\squishlist
\item We propose two novel algorithms for connected components -- {\em (i)} \Hcm,
which provably requires at most $3\log{n}$ rounds with high probability, and at most $2(\nnodes + \nedges)$ communication per round, and {\em (ii)} \Hm, which we prove requires at most $4\log{n}$ rounds on path graphs, and requires $2 \log d$ rounds and $3(\nnodes+\nedges)$ communication per round in practice. (Section~\ref{sec:algos})
\item While \Hcm requires connected components to fit in a single machine's memory, we propose a robust implementation of \Hm that scales with arbitrarily large connected components. We also describe extensions to \Hm for load balancing, and show that on large social network graphs, for which \Hcm runs out of memory, \Hm still works efficiently. (Section~\ref{sec:implementation})
\item We also present two algorithms for single linkage agglomerative clustering using our framework: one using \Ha that provably finishes in $O(\log n)$ map-reduce rounds, and at most $O(n\nnodes + \nedges)$ communication per round, and the other using \Hm that we again conjecture finishes in $O(\log d)$ map-reduce rounds, and at most $O(\nnodes + \nedges)$ communication per round. (Section~\ref{sec:single_linkage}) 
\item We present detailed experimental results evaluating our algorithms for connected components and clustering  and compare them with previously proposed algorithms on multiple real-world datasets. (Section~\ref{sec:experiments})
\squishend

We present related work  in Sec.~\ref{sec:relatedwork}, followed by algorithm and experiment sections, and then conclude in Sec.~\ref{sec:conclusions}.

\section{Related Work}
\label{sec:relatedwork}

\begin{table}[t]
\centering
\begin{tabular}{ |c|c|c| }
  \hline
Name & \# of steps & Communication \\ \hline
   Pegasus~\cite{pegasus} & $O(d)$  &  ${O(\nnodes+\nedges)}$  \\ 
  Zones~\cite{twiddle} & $O(d)$ & ${O(\nnodes+\nedges)}$  \\ 
  L Datalog~\cite{foto11:recursive} & $O(d)$ & $ O(n\nnodes+\nedges)$ \\ 
  NL Datalog~\cite{foto11:recursive} & $O(\log{d})$ & $O(n\nnodes+\nedges)$  \\ 
  PRAM~\cite{Shiloach82,Reif85,gazit91:cc,kargerNP99:cc,karloffSV10:mapreduce}& $O(\log{n})$ &
 shared memory\footnote{\cite{karloffSV10:mapreduce} simulates shared memory by having a reducer per each element} \\ \hline
  {\bf \Hcm} & $3\log{n}$ & $2(\nnodes+\nedges)$ \\  \hline
\end{tabular}
\caption{Complexity comparison with related work: $n= \#$ of nodes in largest component, and $d =$ graph diameter}
\label{table:related}
\vspace{-0.4in}
\end{table}

The problems of finding connected components and undirected $s$-$t$ connectivity (USTCON) are fundamental and very well studied in many distributed settings including PRAM, MapReduce, and BSP. We discuss each of them below.

\subsection{Parallel Random Access Machine (PRAM)}
\label{sec:PRAM}

The PRAM computation model allows several processors to compute in parallel using a common shared memory. PRAM can be classified as CRCW PRAM if concurrent writes to shared memory are permitted, and CREW PRAM if not. 
Although, map-reduce does not have a shared memory, PRAM algorithms are still relevant, due to two reasons: (i) some PRAM algorithms can been ported to map-reduce by case-to-case analyses, and (ii), a  general theoretical result~\cite{karloffSV10:mapreduce} shows that any $O(t)$ CREW PRAM algorithm can be simulated in $O(t)$ map-reduce steps.

For the CRCW PRAM model, Shiloach and Vishkin~\cite{Shiloach82} proposed a  deterministic $O(\log{n})$ algorithm to compute connected components, with $n$ being the size of the largest component. Since then, several other $O(\log{n})$  CRCW algorithms have been proposed in~\cite{gazit91:cc,krishnamurthy94,Reif85}. However, since they require concurrent writes, it is not obvious how to translate them to map-reduce efficiently, as the simulation result of~\cite{karloffSV10:mapreduce} applies only to CREW PRAM.

For the CREW PRAM model, Johnson et. al.~\cite{Johnson97} provided a deterministic $O(\log^{3/2}{n})$ time algorithm, which was subsequently improved to $O(\log{n})$ by Karger et. al.~\cite{kargerNP99:cc}. These algorithms can be simulated in map-reduce using the result of~\cite{karloffSV10:mapreduce}. However, they require computing all nodes at a distance $2$ of each node, which would require $O(n^2)$ communication per map-reduce iteration on a star graph.

Conceptually, our algorithms are most similar to the CRCW PRAM algorithm of Shiloach and Vishkin~\cite{Shiloach82}. That algorithm maintains a connected component as a forest of trees, and repeatedly applies either the operation of pointer doubling (pointing a node to its grand-parent in the tree), or of hooking a tree to another tree.
Krishnamurthy et al \cite{krishnamurthy94} propose a more efficient implementation, similar to map-reduce, by interleaving local computation on local memory, and parallel computation on shared memory. However, pointer doubling and hooking require concurrent writes, which are hard to implement in map-reduce. Our \Hm algorithm does conceptually similar but, slightly different, operations in a single map-reduce step.

\subsection{Map-reduce Model}
Google's map-reduce lecture series describes an iterative approach for computing connected components. In each iteration a series of map-reduce steps are used to find and include all nodes adjacent to current connected components. The number of iterations required for this method, and many of its improvements~\cite{twiddle,pegasus,mpi}, is $O(d)$ where $d$ is the diameter of the largest connected component. 
These techniques do not scale well for large diameter graphs (such as graphical models where edges represent correlations between variables). Even for moderate diameter graphs~(with $d=20$), our techniques outperform the $O(d)$ techniques, as shown in the experiments.

Afrati et al \cite{foto11:recursive} propose map-reduce algorithms for computing transitive closure of a graph --  a relation containing tuples of pairs of nodes that are in the same connected component. These techniques have a larger communication per iteration as the transitive closure relation itself is quadratic in the size of largest component. Recently, Seidl et al \cite{conf/pkdd/SeidlBF12} have independently proposed map-reduce algorithms similar to ours, including the use of secondary sorting. However, they do not show the $O(\log{n})$ bound on the number of map-reduce rounds.

Table~\ref{table:related} summarizes the related work comparison and shows that our \Hcm algorithm is the first map-reduce technique with logarithmic number of iterations and linear communication per iteration.



\subsection{Bulk Synchronous Parallel (BSP)}
In the BSP paradigm, computation is done in parallel by processors in between a series of synchronized point-to-point communication steps. The BSP paradigm is used by recent distributed graph processing systems like Pregel \cite{Malewicz2010:Pregel} and Giraph \cite{Ching2010:Giraph}. BSP is generally considered more efficient for graph processing than map-reduce as it has less setup and overhead costs for each new iteration. While the algorithmic improvements of reducing number of iterations presented in this paper are applicable to BSP as well, these improvements are of less significance in BSP due to lower overhead of additional iterations.

However, we show that BSP does not necessarily dominate map-reduce for large-scale graph processing~(and thus our algorithmic improvements for map-reduce are still relevant and important). We show this by running an interesting experiment in shared grids having congested environments in Sec.~\ref{sec:bspexp}.
%

The experiment shows that in congested clusters, map-reduce can have a better latency than BSP, since in the latter one needs to acquire and hold machines with a combined memory larger than the graph size. For instance, consider a graph with a billion nodes and ten billion edges. Suppose each node is associated with a state of 256 bytes  (e.g., the contents of a web page, or recent updates by a user in a social network, etc.). Then the total memory required would be about 256 GB, say, 256 machines with 1G RAM. In a congested grid waiting for 256 machines could take much longer than running a map-reduce job, since the map-reduce jobs can work with a smaller number of mappers and reducers (say 50-100), and switch in between different MR jobs in the congested environment.

\section{Connected Components on Map-Reduce}
\label{sec:algos}

\begin{algorithm}[t]
\caption{General Map Reduce Algorithm}
\begin{algorithmic}[1]
\STATE {\bf Input:} A graph $G = (V,E)$, \newline
{\em hashing} function $\h$ \newline
{\em merging} function $\Md$, and \newline
{\em export} function \export
\STATE {\bf Output:} A set of connected components $\C \subset 2^V$
\STATE Either Initialize $\TC_v = \{v\}$ Or $\TC_v = \{v\} \cup \nbrs(v)$ depending on the algorithm.
 \REPEAT
  \STATE {\bf mapper for node $v$:}
  \STATE  Compute $\h(\TC_v)$, which is a collection of key-value pairs $(u, C_u)$ for $u \in \TC_v$.
  \STATE  Emit all $(u, C_u) \in \h(\TC_v)$.
  \STATE {\bf reducer for node $v$:}
  \STATE Let $\{C_v^{(1)}, \ldots, C_v^{(K)}\}$ denote the set of values received from different mappers.
  \STATE Set $\TC_v \leftarrow \Md(\{C_v^{(1)}, \ldots, C_v^{(K)}\})$
 \UNTIL {$\TC_v$ does not change for all $v$}
 \STATE Return $\C = \mbox{\export}(\cup_v \{ \TC_v \})$
\end{algorithmic}
\label{algo:distributed}
\end{algorithm}

In this section, we present map-reduce algorithms for computing connected components. All our algorithms are instantiations of a general map-reduce framework (Algorithm \ref{algo:distributed}), which is parameterized by two functions -- a {\em hashing} function  $\h$, and a  {\em merging} function $\Md$ (see line~1 of Algorithm~\ref{algo:distributed}). Different choices for $\h$ and $\Md$~(listed in Table~\ref{hashlist}) result in algorithms having very different complexity.

Our algorithm framework maintains a tuple (key, value) for each node $v$ of the graph -- key is the node identifier $v$, and the value is a cluster of nodes, denoted $\TC_v$. The value $\TC_v$ is initialized as either containing only the node $v$, or containing $v$ and all its neighbors $\nbrs(v)$ in $G$, depending on the algorithm (see line~3 of Algorithm~\ref{algo:distributed}). The framework updates $\TC_v$ through multiple mapreduce iterations.

In the map stage of each iteration, the mapper for a key $v$ applies the hashing function $\h$ on the value $\TC_v$ to emit a set of key-value pairs $(u ,C_u)$, one for every node $u$ appearing in $\TC_v$ (see lines~6-7). The choice of hashing function governs the behavior of the algorithm, and we will discuss different instantiations shortly. In the reduce stage, each reducer for a key $v$ aggregates tuples  $(v,C_v^{(1)}), \ldots, (v,C_v^{(K)})$ emitted by different mappers. The reducer applies the merging function $\Md$ over $C_v^{(i)}$ to compute a new value $\TC_v$ (see lines 9-10). This process is repeated until there is no change to any of the clusters $\TC_v$ (see line 11).  Finally, an appropriate \export function computes the connected components $\C$ from the final clusters $\TC_v$ using one map-reduce round.

\begin{table}[t]
\centering
\begin{tabular}{|l|}
\hline
\begin{minipage}[t]{\linewidth}
\noindent {\bf \Hma} emits $(v, \TC_v)$, and $(u, \{v_{min}\})$ for all nodes $u \in \nbrs(v)$.\\
\noindent {\bf \Ha} emits $(u,\TC_v)$ for all nodes $u \in \TC_v$. \\
\noindent {\bf \Hm} emits $(v_{\min}, \TC_v)$, and $(u,\{v_{min}\})$ for all nodes $u \in \TC_v$.\\
\noindent {\bf \Hcm} computes $\core{v}$, the set of nodes in $\TC_v$ not less than $v$. It emits $(v_{\min}, \core{v})$, and $(u,\{v_{min}\})$ for all nodes $u \in \core{v}$
\end{minipage}
\\ \hline
\end{tabular}
\caption{Hashing Functions:  Each strategy describes the key-value pairs emitted by mapper with input key $v$ and value $\TC_v$~($v_{\min}$ denotes smallest node in $\TC_v$)}
\label{hashlist}
\vspace{-0.3in}
\end{table}

\begin{table}[t]
\centering
{\small
\begin{tabular}{ |c|c|c| }
  \hline
{\bf Algorithm} & {\bf MR Rounds} & {\bf Communication} \\
& & {\bf (per MR step)} \\
  \hline
  \Hma~\cite{pegasus} & $d$ & $O(\nnodes + \nedges)$   \\
  \Ha & $\log{d}$ & $O(n\nnodes + \nedges)$ \\
\Hm & $O(\log{n})^\star$ & $O(\log{n}\nnodes + \nedges)^\star$ \\
  \Hcm & $3\log{n}$ & $2(\nnodes+\nedges)$ \\	
 \hline
\end{tabular}}
\caption{Complexity of Different Algorithms~($^\star$ denotes results hold for only path graphs)}
\label{hashlistcomplexity}
\vspace{-0.3in}
\end{table}

\paragraph{Hash Functions} We describe four hashing strategies in Table~\ref{hashlist} and their complexities in Table~\ref{hashlistcomplexity}. The first one, denoted \Hma, was used in~\cite{pegasus}. In the mapper for key $v$, \Hma emits key-value pairs $(v,\TC_v)$ and $(u, \{v_{min}\})$ for all nodes $u \in \nbrs(v)$. In other words, it sends the entire cluster $\TC_v$ to reducer $v$ again, and sends only the minimum node $v_{\min}$ of the cluster $\TC_v$ to all reducers for nodes $u \in \nbrs(v)$. Thus communication is low, but so is rate of convergence, as information spreads only by propagating the minimum node.

On the other hand, \Ha emits key-value pairs $(u,\TC_v)$ for all nodes $u \in \TC_v$. In other words, it sends the cluster $\TC_v$ to all reducers $u \in \TC_v$.  Hence if clusters $\TC_v$ and $\TC_v'$ overlap on some node $u$, they will both be sent to reducer of $u$, where they can be merged, resulting in a faster convergence. But, sending the entire cluster $\TC_v$ to all reducers $u \in \TC_v$ results in large quadratic communication cost. To overcome this, \Hm sends the entire cluster $\TC_v$ to only one reducer $v_{\min}$, while other reducers are just sent $\{ v_{\min} \}$.  This decreases the communication cost drastically, while still achieving fast convergence. Finally, the best theoretical complexity bounds can be shown for \Hcm, which sends out a smaller subset $\core{v}$ of $\TC_v$.  We look at how these functions are used in specific algorithms next.

\subsection{\Hma Algorithm}
The \Hma algorithm is a version of the Pegasus algorithm~\cite{pegasus}.\footnote{\cite{pegasus} has additional optimizations that do not change the asymptotic complexity. We do not describe them here.} In this algorithm each node $v$ is associated with a label $v_{min}$ (i.e., $\TC_v$ is a singleton set $\{v_{min}\}$) which corresponds to the smallest id amongst nodes that $v$ knows are in its connected component. Initially $v_{\min} = v$ and so $\TC_v =  \{v \}$. It then uses \Hma hashing function to propagate its label $v_{min}$ in $\TC_v$ to all reducers $u \in \nbrs(v)$ in every round. On receiving the messages, the merging function $\Md$ computes the smallest node $v^{new}_{\min}$ amongst the incoming messages and sets $\TC_v = \{ v^{new}_{\min} \}$. Thus a node adopts the minimum label found in its neighborhood as its own label. On convergence, nodes that have the same label are in the same connected component. Finally, the connected components are computed by the following \export function: return sets of nodes grouped by their label.

 \begin{theorem}[\Hma~\cite{pegasus}] \label{thm:hma}
Algorithm \Hma correctly computes the connected components of  $G = (V,E)$ using $O(\nnodes + \nedges)$ communication and $O(d)$ map-reduce rounds.
 \end{theorem}
 \eat{
 \begin{proof}(sketch)
 First, \Hma correctly computes the connected components since any node that is connected to the smallest node $u_{\min}$ via a shortest path of length $p$ gets $u_{\min}$ in exactly $p$ steps. This also shows that the algorithm requires $O(d)$ map-reduce steps, since in each step the smallest node travels one edge. Finally, the communication cost is $O(\nnodes+\nedges)$ since a breadth first search is simulated.
 \end{proof}
}

 \subsection{\Ha Algorithm}
 The \Ha algorithm initializes each cluster $\TC_v = \{v\} \cup \{ \nbrs(v) \}$. Then it uses \Ha hashing function to send the entire cluster $\TC_v$ to all reducers $u \in \TC_v$.  On receiving the messages, merge function $\Md$ updates the cluster by taking the union of all the clusters received by the node. More formally, if the reducer at $v$ receives clusters $C_v^{(1)}, \ldots, C_v^{(K)}$, then $\TC_v$ is updated to $\cup_{i = 1}^K C_v^{(i)}$.

We can show that after $\log d$ map-reduce rounds, for every $v$, $\TC_v$ contains all the nodes in the connected component containing $v$. Hence, the \export function just returns the distinct sets in $\TC$ (using one map-reduce step).

 \begin{theorem}[\Ha]
Algorithm \Ha correctly computes the connected components of  $G = (V,E)$ using $O(n\nnodes + \nedges)$ communication per round and $\log d$ map-reduce rounds, where $n$ is the size of the largest component and $d$ the diameter of $G$.
 \end{theorem}
 \begin{proof}
 We can show using induction that after $k$ map-reduce steps, every node $u$ that is at a distance $\leq 2^k$ from $v$ is contained in $\TC_v$. Initially this is true, since all neighbors are part of $\TC_v$. Again, for the $k+1^{st}$ step,  $u \in \TC_w$ for some $w$ and $w \in \TC_v$ such that distance between $u, w$ and $w, v$ is at most $2^k$. Hence, for every node $u$ at a distance at most $2^{k+1}$ from $v$, $u \in \TC_v$ after $k+1$ steps. Proof for communication complexity follows from the fact that each node is replicated at most $n$ times.
 \end{proof}


\subsection{\Hm Algorithm}
While the \Ha algorithm computes the connected components in a smaller number of map-reduce steps than \Hma, the size of the intermediate data (and hence the communication) can become prohibitively large for even sparse graphs having large connected components. We now present \Hm, a variation on \Ha, that we show finishes in at most $4\log{n}$ steps for path graphs. We also show that in practice it takes at most $2\log{d}$ rounds and linear communication cost per round (see Section~\ref{sec:experiments}),where $d$ is the diameter of the graph.

The \Hm  algorithm initializes each cluster $\TC_v = \{v\} \cup \{ \nbrs(v) \}$. Then it uses \Hm hash function to send the entire cluster $\TC_v$ to reducer $v_{min}$, where $v_{min}$ is the smallest node in the cluster $\TC_v$, and $\{ v_{min} \}$ to all reducers $u \in \TC_v$. The merging function $\Md$ works exactly like in \Ha: $\TC_v$ is the union of all the nodes appearing in the received messages. We explain how this algorithm works by an example.\\

\begin{example} Consider an intermediate step where clusters  $C_1 = \{1,2,4\}$ and $C_5 = \{3,4,5\}$ have been associated with keys $1$ and $5$. We will show how these clusters are merged in both \Ha and \Hm algorithms.

 In the \Ha scheme, the mapper at $1$ sends the entire cluster $C_1$ to reducers $1$, $2$, and $4$, while mapper at $5$ sends $C_5$ to reducers $3$, $4$, and $5$. Therefore, on reducer $4$, the entire cluster $C_4 = \{1,2,3,4,5\}$ is computed by the merge function. In the next step, this cluster $C_4$ is sent to all the five reducers.

In the \Hm\ scheme, the mapper at $1$ sends $C_1$ to reducer $1$, and $\{1\}$ to reducer $2$ and $4$. Similarly, the mapper at $5$ sends $C_5$  to reducer $3$, and $\{3\}$ to reducer $4$ and $5$. So reducer $4$ gets $\{1\}$ and $\{3\}$, and therefore computes the cluster $C_4 = \{1,3\}$ using the merge function.

Now, in the second round, the mapper at $4$, has $1$ as the minimum node of the cluster $C_4 = \{1,3\}$. Thus, it sends $\{1\}$ to reducer $3$, which already has the cluster $C_2 = \{3,4,5\}$. Thus after the second round, the cluster $C_3 = \{1,3,4,5\}$ is formed on reducer $3$. Since $1$ is the minimum for $C_3$, the mapper at $3$ sends $C_3$  to reducer $1$ in the third round. Hence after the end of third round, reducer $1$ gets the entire cluster $\{1,2,3,4,5\}$.

Note in this example that \Hm\ required three map-reduce steps; however, the intermediate data transmitted is lower since entire clusters $C_1$ and $C_2$ were only sent to their minimum element's reducer ($1$ and $3$, resp).\\
\end{example}

As the example above shows, unlike \Ha, at the end of \Hm, all reducers $v$ are not guaranteed to contain in $\TC_v$ the connected component they are part of. In fact, we can show that the reducer at $v_{min}$ contains all the nodes in that component, where $v_{min}$ is the smallest node in a connected component. For other nodes $v$, $\TC_v = \{v_{\min}\}$. Hence, \export outputs only those $\TC_v$ such that $v$ is the smallest node in $\TC_v$.

\begin{theorem}[\Hm Correctness]\label{lemma:hm}
  At the end of algorithm \Hm, $\TC_v$ satisfies the following property: If $v_{\min}$ is the smallest node of a connected component $C$, then $\TC_{v_{\min}} = C$. For all other nodes $v$, $\TC_v = \{v_{\min}\}$.
 \end{theorem}
 \begin{proof}
 Consider any node $v$ such that $\TC_v$ contains $v_{\min}$. Then in the next step, mapper at $v$ sends $\TC_v$ to $v_{\min}$, and only $\{ v_{\min} \}$ to $v$. After this iteration, $\TC_v$ will always have $v_{min}$ as the minimum node, and the mapper at $v$ will always send its cluster $\TC_v$ to $v_{min}$. Now at some point of time, all nodes $v$ in the connected component $C$ will have $v_{\min} \in \TC_{v}$ (this follows from the fact that min will propagate at least one hop in every iteration just like in Theorem~\ref{thm:hma}). Thus, every mapper for node $v$ sends its final cluster to $v_{\min}$, and only retains $v_{\min}$. Thus at convergence $\TC_{v_{min}} = C$ and $\TC_{v} = \{ v_{min} \}$.
 \end{proof}

 \begin{theorem}[\Hm Communication]
\ \\ Algorithm takes $O(k \cdot (\nnodes + \nedges))$ expected communication per round, where $k$ is the total number of rounds. Here expectation is over the random choices of the node ordering.
 \end{theorem}
 \begin{proof}
Note that the total communication in any step equals the total size of all $\TC_v$ in the next round. Let $n_k$ denote the size of this intermediate after $k$ rounds. That is, $n_k = \sum_v \TC_v$. We show by induction that $n_k = O( k \cdot (\nnodes + \nedges))$.

First, $n_0 = \sum_v \TC^{0}_v \leq 2(\nnodes + \nedges)$, since each node contains itself and all its neighbors. In each subsequent round, a node $v$ is present in  $\TC_u$, for all $u \in \TC_v$. Then $v$ is sent to a different cluster in one of two ways:
\squishlist
\item If $v$ is the smallest node in $\TC_u$, then $v$ is sent to all nodes in $\TC_u$. Due to this, $v$ gets replicated to $|\TC_u|$ different clusters. However, this happens with probability $1/|\TC_u|$.
\item If $v$ is not the smallest node, then $v$ is sent to the smallest node of $\TC_u$. This happens with probability $1 - 1/|\TC_u|$. Moreover, once $v$ is not the smallest for a cluster, it will never become the smallest node; hence it will never be replicated more that once.
\squishend
From the above two facts, on expectation after one round, the node $v$ is sent to  $s_1 = |\TC^{0}_v|$ clusters as the smallest node and to $m_1 = |\TC^{0}_v|$ clusters as not the smallest node. After two rounds, the node $v$ is additionally sent to $s_2 = |\TC^{0}_v|$, $m_2 = |\TC^{0}_v|$, in addition to the $m_1$ clusters. Therefore, after $k$ rounds,  $n_k = O(k \cdot (\nnodes + \nedges))$.
 \end{proof}

Next we show that on a path graph, \Hm finishes in $4 \log{n}$. The proof is rather long, and due to space constraints appears in Sec.~\ref{sec:Hm} of the Appendix.~

\begin{theorem}[\Hm Rounds]\label{thm:Hm:rounds}
 Let $G=(V,E)$ be a path graph (i.e. a tree with only nodes of degree 2 or 1). Then, \Hm correctly computes the connected component of  $G = (V,E)$ in $4\log{n}$ map-reduce rounds.
\end{theorem}

Although, Theorem~\ref{thm:Hm:rounds} works only for path graphs, we conjecture that \Hm  finishes in $2 \log d$ rounds on all inputs, with  $O(\nnodes + \nedges)$ communication per round. Our experiments (Sec.~\ref{sec:experiments}) seem to validate this conjecture.



\subsection{\Hcm Algorithm}
Now we describe the \Hcm algorithm that has the best theoretical bounds: $3\log{n}$ map-reduce rounds with high probability and 2(\nnodes+\nedges) communication complexity per round in the worst-case. In \Hcm algorithm, the clusters $\TC_v$ are again initialized as $\{v\}$. Then \Hcm algorithm runs two rounds using \Hma hash function, followed by a round using \Hcm hash function, and keeps on repeating these three rounds until convergence.

In a round using \Hma hash function, the entire cluster $\TC_v$ is sent to reducer $v$ and $v_{\min}$ to all reducers $u \in \nbrs(v)$. For, the merging function $\Md$ on machine $m(v)$, the algorithm first computes the minimum node among all incoming messages, and then adds it to the message $\TC(v)$ received from $m(v)$ itself. More formally, say $v^{new}_{\min}$ is the smallest nodes among all the messages received by $u$, then $\TC_{new}(v)$ is updated to $\{v^{new}_{\min}\} \cup \{\TC(v)\}$.

In a round using \Hcm hash function, the set $\core{v}$ is computed as all nodes in $\TC_v$ not less than $v$. This set is sent to reducer $v_{\min}$, where $v_{\min}$ is the smallest node in $\TC(v)$, and $\{ v_{\min} \}$ is sent to all reducers $u \in \core{v}$. The merging function $\Md$ works exactly like in \Ha: $\TC(v)$ is the union of all the nodes appearing in the received messages. We explain this process by the following example.\\

\begin{example} Consider a path graph with $n$ edges $(1,2)$, $(2,3)$, $(3,4)$, and so on. We will now show three rounds of \Hcm.

In \Hcm algorithm, the clusters are initialized as $\TC_i = \{ i \}$ for $i \in [1,n]$. In the first round, the \Hma function will send $\{ i \}$ to reducers $i-1$, $i$, and $i+1$. So each reducer $i$ will receive messages $\{i-1\}$,  $\{i\}$ and $\{i+1\}$, and aggregation function will add the incoming minimum, $i-1$, to the previous $\TC_i = \{ i \}$.

Thus in the second round, the clusters are $\TC_1 = \{1\}$ and $\TC_i = \{ i-1, i\}$ for $i \in [2,n]$. Again \Hma will send the minimum node $\{ i-1 \}$ of $\TC_i$ to reducers $i-1$, $i$, and $i+1$.  Again merging function would be used. At the end of second step, the clusters are $\TC_1 = \{1\}$, $\TC_2 = \{1,2\}$, $\TC_i = \{i-2,i-1,i,\}$ for $i \in [3,n]$.

In the third round, \Hcm will be used. This is where interesting behavior is seen. Mapper at $2$ will send its $\core(2) = \{2\}$ to reducer $1$. Mapper at $3$ will send its $\core(3) = \{3\}$ to reducer $1$. Note that $\core(3)$ does not include $2$ even though it appears in $\TC_3$ as $2 < 3$. Thus we save on sending redundant messages from mapper $3$ to reducer $1$ as $2$ has been sent to reducer $1$ from mapper $2$. Similarly, mapper at $4$ sends $\core(4) = \{4\}$ to reducer $2$, and mapper $5$ sends $\core(5) = \{5\}$ to reducer $3$, etc. Thus we get the sets, $\TC_1 = \{1,2,3\}$, $\TC_2 = \{1,2,4\}$, $\TC_3 = \{1,3,6\}$, $\TC_4 = \{4,5,6\}$, and so on.\\
\end{example}

The analysis of the \Hcm algorithm relies on the following lemma.
\begin{lemma} \label{lemma:hgm} Let $v_{\min}$ be any node. Denote $\set(v_{\min})$ the set of all nodes $v$ for which $v_{\min}$ is the smallest node in $\TC_v$ after \Hcm algorithm converges. Then $\set(v_{\min})$ is precisely the set $\core(v_{\min})$.
\end{lemma}
Note that in the above example, after 3 rounds of \Hcm, $\set(2)$ is $\{2,4\}$ and $\core(2)$ is also $\{2,4\}$.

We now analyze the performance of this algorithm. The proof is based on techniques introduced in~\cite{gazit91:cc,kargerNP99:cc}, and omitted here due to lack of space. The proof appears in the Appendix.

 \begin{theorem}[Complexity]\label{thm:hgm}
 Algorithm \Hcm correctly computes the connected components of  $G = (V,E)$ in expected $3\log{n}$ map-reduce rounds (expectation is over the random choices of the node ordering) with  $2(\nnodes + \nedges)$ communication per round in the worst case.
 \end{theorem}

\section{Scaling the \Hm Algorithm}\label{sec:implementation}

\eat{
\begin{algorithm}[t]
\caption{Hybrid Algorithm}
\begin{algorithmic}[1]
\STATE {\bf Input:} A graph $G = (V,E)$, Threshold $\tau$
\STATE {\bf Output:} Connected components of $G$
\STATE Initialize all nodes to mode 1.
\STATE Run \Hm~(\Hma) on mode 1~(2) nodes..
\IF{$|\TC(v)| > \tau$ or $m(v)$ gets mode 2 signal}
	\STATE Change $v$ to mode 2.
        \STATE Send mode 2 signal to all nodes $u \in \TC(v)$.
\ENDIF
\STATE Repeat from Step 4 until connect components do not change.
\end{algorithmic}
\label{algo:hybrid}
\end{algorithm}
}

\Hm and \Hcm complete in less number of rounds than \Hma, but as currently described, they require that every connected component of the graph fit in memory of a single reducer. We now describe a more robust implementation for \Hm, which allows handling arbitrarily large connected components. We also describe an extension to do load balancing. Using this implementation, we show in Section~\ref{sec:experiments} examples of social network graphs that have small diameter and extremely large connected components, for which \Hcm runs out of memory, but \Hm still works efficiently. 

\subsection{Large Connected Components}
We address the problem of larger than memory connected components, by using secondary sorting in map-reduce, which allows a reducer to receive values for each key in a sorted order. Note this sorting is generally done in map-reduce to keep the keys in a reducer in sorted order, and can be extended to sort values as well, at no extra cost, using composite keys and custom partitioning~\cite{Lin2010}.

To use secondary sorting, we represent a connected component as follows: if $\C_v$ is the cluster at node $v$, then we represent $\C_v$ as a graph with an edge from $v$ to each of the node in $\C_v$. Recall that each iteration of \Hm is as follows: for hashing, denoting $v_{min}$ as the min node in $\C_v$, the mapper at $v$ sends $\C_v$ to reducer $v_{min}$, and $\{v_{min}\}$ to all reducers $u$ in $\C_v$. For merging, we take the union of all incoming clusters at reducer $v$. 

\Hm can be implemented in a single map-reduce step. The hashing step is implemented by emitting in the mapper, key-value pairs, with key as $v_{min}$, and values as each of the nodes in $\C_v$, and conversely, with key as each node in $\C_v$, and $v_{min}$ as the value. The merging step is implemented by collecting all the values for a key $v$ and removing duplicates.

To remove duplicates without loading all values in the memory, we use secondary sorting to guarantee that the values for a key are provided to the reducer in a sorted order. Then the reducer can just make a single pass through the values to remove duplicates, without ever loading all of them in memory, since duplicates are guaranteed to occur adjacent to each other. Furthermore, computing the minimum node $v_{min}$ is also trivial as it is simply the first value in the sorted order. 

\subsection{Load Balancing Problem}
Even though \Hm can handle arbitrarily large graphs without failure (unlike \Hcm), it can still suffer from data skew problems if some connected components are large, while others are small. We handle this problem by tweaking the algorithm as follows. If a cluster $\C_v$ at machine $v$ is larger than a predefined threshold, we send all nodes $u \leq v$ to reducer $v_{min}$ and $\{v_{min}\}$ to all reducers $u \leq v$, as done in \Hm. However, for nodes $u>v$,  we send them to reducer $v$ and $\{v\}$ to reducer $u$. This ensures that reducer $v_{min}$ does not receive too many nodes, and some of the nodes go to reducer $v$ instead, ensuring balanced load.

This modified \Hm is guaranteed to converge in at most the number of steps as the standard \Hm converges. However, at convergence, all nodes in a connected component are not guaranteed to have the minimum node $v_{min}$ of the connected component. In fact, they can have as their minimum, a node $v$ if the cluster at $v$ was bigger than the specified threshold. We can then run standard \Hm, on the modified graph over nodes that correspond to cluster ids, and get the final output. Note that this increases the number of rounds by at most 2, as after load-balanced \Hm converges, we use the standard \Hm. 

\begin{example} If the specified threshold is $1$, then the modified algorithm converges in exactly one step, returning clusters equal to one-hop neighbors. If the specified threshold is $\infty$, then the modified algorithm converges to the same output as the standard one, i.e. it returns connected components. If the specified threshold is somewhere in between~(for our experiments we choose it to 100,000 nodes), then the output clusters are subsets of connected components, for which no cluster is larger than the threshold. 
\end{example}

\newcommand{\Slc}{Single Linkage Clustering\xspace}
\newcommand{\slc}{single linkage clustering\xspace}

\section{Single Linkage Agglomerative Clustering}
\label{sec:single_linkage}
To the best of our knowledge, no map-reduce implementation exists for \slc that completes in $o(n)$ map-reduce steps, where $n$ is the size of the largest cluster. We now present two map-reduce implementations for the same, one using \Ha that completes in $O(\log n)$ rounds, and another using \Hm that we conjecture to finish in $O(\log d)$ rounds.

For clustering, we take as input a weighted graph denoted as  $G =
(V,E, w)$, where $w: E \rightarrow [0,1]$ is a weight function on
edges. An output cluster $C$ is any set of nodes, and a clustering
$\cal C$ of the graph is any set of clusters such that each node
belongs to exactly one cluster in $\C$.

\begin{algorithm}[t]
\caption{Centralized \slc}
\begin{algorithmic}[1]
\STATE {\bf Input:} Weighted graph $G = (V,E, w)$, \newline
stopping criterion $\Stop$.
\STATE {\bf Output:} A clustering $\C \subseteq 2^V$.
\STATE Initialize clustering $\C \leftarrow  \{\{v\} | v\in V\}$;
 \REPEAT
 \STATE Find the closest pair of clusters $C_1,C_2$ in $\C$~(as per $d$);
 \STATE Update $\C \leftarrow  \C -\{C_1,C_2\} \cup \{ C_1 \cup C_2) \}$;
\UNTIL{$\C$ does not change or $\Stop(\C)$ is true}
\STATE Return $\C$
\end{algorithmic}
\label{algo:central}
\end{algorithm}

\begin{algorithm}[t]
\caption{Distributed \slc}
\begin{algorithmic}[1]
\STATE {\bf Input:} Weighted graph $G = (V,E, w)$, \newline
stopping criterion $\Stop$.
\STATE {\bf Output:} A clustering $\C \subseteq 2^V$.
\STATE Initialize $\C = \{ \{v\} \cup \nbrs(v) | v \in V \}$.
 \REPEAT
  \STATE {\bf Map:} Use \Ha or \Hm to hash clusters.
  \STATE {\bf Reduce:} Merge incoming clusters.
 \UNTIL{$\C$ does not change or $\Stop(\C)$ is true}
 \STATE Split clusters in $\C$ merged incorrectly in the final iteration.
 \STATE Return $\C$
\end{algorithmic}
\label{algo:distributed:clustering}
\end{algorithm}

\subsection{Centralized Algorithm:} 
Algorithm~\ref{algo:central} shows the typical bottom up centralized algorithm for \slc. Initially, each node is its own cluster. Define the distance between two clusters $C_1, C_2$ to be the minimum weight of an edge between the two clusters; i.e.,
$$d(C_1, C_2) \ = \ \min_{e = (u,v), u \in C_1, v \in C_2} w(e)$$
In each step, the algorithm picks  the two closest clusters and merges
them by taking their union. The algorithm terminates either when the
clustering does not change, or when a {\em stopping} condition,
$\Stop$, is reached.  Typical stopping conditions are {\em threshold}
stopping, where the clustering stops when the closest distance between
any pair of clusters is larger than a threshold, and {\em cluster
  size} stopping condition, where the clustering stops when the merged
cluster in the most recent step is too large. 

Next we describe a map-reduce algorithm that simulates the centralized
algorithm, i.e., outputs the same clustering. If there are two edges in
the graph having the exact same weight, then \slc might not be unique,
making it impossible to prove our claim. Thus, we assume that ties
have been broken arbitrarily, perhaps, by perturbing the weights slightly, and
thus no two edges in the graph have the same weight. Note that this step is required
only for simplifying our proofs, but not for the correctness of our algorithm.

\subsection{Map-Reduce Algorithm} 
Our Map-Reduce algorithm is shown in
Algorithm~\ref{algo:distributed:clustering}. Intuitively, we can
compute single-linkage clustering by first computing the connected
components~(since no cluster can lie across multiple connected
components), and then splitting the connected components into
appropriate clusters. Thus,
Algorithm~\ref{algo:distributed:clustering} has the same map-reduce
steps as Algorithm~\ref{algo:distributed}, and is implemented either using \Ha or \Hm.

However, in general, clusters, defined by the stopping criteria, $\Stop$, may be much smaller than the connected components. In the extreme case, the graph might be just one giant connected component, but the clusters are often small enough that they individually fit in the memory of a single machine. Thus we need a way to check and stop execution as soon as clusters have been computed.  We do this by evaluating $\Stop(\C)$ after each iteration of map-reduce. If $\Stop(\C)$ is false, a new iteration of map-reduce clustering is started. If $\Stop(\C)$ is true, then we stop iterations. 

While the central algorithm can implement any stopping
condition, checking an arbitrary predicate in a distributed setting can be
difficult. Furthermore, while the central algorithm merges one cluster at a time, and then
evaluates the stopping condition, the distributed algorithm evaluates
stopping condition only at the end of a map-reduce iteration. This
means that some reducers can merge clusters incorrectly in the last
map-reduce iteration. We describe next how to stop the map-reduce clustering
algorithm, and split incorrectly merged clusters.

\subsection{Stopping and Splitting Clusters}
\label{sec:slc_stop}

It is difficult to evaluate an arbitrary stopping predicate in a distributed fashion using map-reduce. We restrict our attention to a restricted yet frequently used class of local monotonic stopping criteria, which is defined below.

\begin{definition}[Monotonic Criterion]  $\Stop$ is monotone if for every clusterings $\cal C$, ${\cal C}'$, if ${\cal C}'$ refines $\cal C$~(i.e, $\forall C \in {\cal C} \Rightarrow \exists C' \in {\cal C}',  ~C \subseteq C')$, then $\Stop({\cal C}) = 1 \Rightarrow \Stop({\cal C}') = 1$.
\end{definition}

Thus monotonicity implies that stopping predicate continues to remain
true if some clusters are made smaller. Virtually every stopping
criterion used in practice is monotonic. Next we define the assumption
of locality, which states that stopping criterion can be evaluated
locally on each cluster individually.

\begin{definition}[Local Criterion]   $\Stop$ is local if there exists a function $\Stop_{local}: 2^{V} \rightarrow \{0,1\}$ such that $\Stop({\cal C}) = 1$ iff $\Stop_{local}(C) = 1$ for all $C \in \cal C$.
\end{definition}

Examples of local stopping criteria include distance-threshold (stop
merging clusters if their distance becomes too large) and
size-threshold (stop merging if the size of a cluster becomes too
large). Example of {\em non}-local stopping criterion is to stop when
the total number of clusters becomes too high.

If the stopping condition is local and monotonic, then we can compute it
efficiently in a single map-reduce step. To explain how, we first
define some notations. Given a cluster $C \subseteq V$, denote $G_C$ the subgraph of $G$ induced over nodes
$C$. Since $C$ is a cluster, we know $G_C$ is connected. We denote
$tree(C)$ as the\footnote{The tree is unique because of unique edge
  weights} minimum weight spanning tree of $G_C$, and $split(C)$ as
the pair of clusters $C_L,C_R$ obtained by removing the edge with the
maximum weight in $tree(C)$. Intuitively, $C_L$ and $C_R$ are the
clusters that get merged to get $C$ in the centralized \slc algorithm. Finally,
denote $\nbrs(C)$ the set of clusters closest to $C$ by the
distance metric $d$, i.e. if $C_1 \in \nbrs(C)$, then for every other cluster $C_2$, $d(C,C_2) > d(C,C_1)$. 

We also define the notion of core and minimal core decomposition as follows. 
\begin{definition}[Core] \label{def:core} A singleton cluster is always a
  core. Furthermore, any cluster $C \subseteq V$ is a core if its
  split $C_L,C_R$ are both cores and closest to each other, i.e. $C_L \in \nbrs(C_R)$ and $C_R \in  \nbrs(C_L)$.
\end{definition}
\begin{definition}[Minimal core decomposition] Given a cluster $C$ its minimal core decomposition,  $MCD(C)$, is a set of cores $\{C_1, C_2, \ldots, C_l\}$ such that $\cup_{i} C_i = C$ and for every core $C' \subseteq C$ there exists a core $C_i$ in the decomposition for which $C' \subseteq C_i$.
\end{definition}

Intuitively, a cluster $C$ is a core, if it is a valid, i.e., it is a subset of
some cluster $C'$ in the output of the centralized \slc algorithm, and
$MCD(C)$ finds the largest cores in $C$, i.e. cores that cannot be merged with any other node in $C$ and still be cores.

\begin{algorithm}[t]
\caption{Minimal Core Decomposition $MCD$}
\begin{algorithmic}[1]
\STATE {\bf Input:} Cluster $C \subseteq V$.
\STATE {\bf Output:} A set of cores $\{C_1, C_2, \ldots, C_l\}$ corresponding to $MCD(C)$.
\STATE If $C$ is a core, return $\{ C \}$.
\STATE Construct the spanning tree $T_C$ of $C$, and compute $C_L,C_R$
to be the cluster split of $C$.
\STATE Recursively compute $MCD(C_L)$ and $MCD(C_R)$.
 \STATE Return $MCD(C_L) \cup MCD(C_R)$
\end{algorithmic}
\label{algo:mcd}
\end{algorithm}

\begin{algorithm}[t]
\caption{Stopping Algorithm}
\begin{algorithmic}[1]
\STATE {\bf Input:} Stopping predicate $\Stop$, Clustering $\C$.
\STATE {\bf Output:} $\Stop(\C)$
\STATE For each cluster $C \in \C$, compute $MCD(C)$. (performed in
reduce of calling of Algo.~\ref{algo:distributed:clustering})
\STATE {\bf Map:} Run \Ha on cores, i.e, hash each core $C_i \in
MCD(C)$ to all machines $m(u)$ for $u \in C_i$   
\STATE {\bf Reducer for node v:}  Of all incoming cores, pick
the largest core, say, $C_v$, and compute $\Stop_{local} (C_v)$.
\STATE Return $\wedge_{v \in V} \Stop_{local} (C_v)$ 
\end{algorithmic}
\label{algo:stop}
\end{algorithm}

\begin{algorithm}[t]
\caption{Recursive Splitting Algorithm Split}
\begin{algorithmic}[1]
\STATE {\bf Input:} Incorrectly merged cluster $C$ w.r.t $\Stop_{local}$.
\STATE {\bf Output:} Set $S$ of correctly split clusters in $C$.
\STATE Initialize $S = \{\}$.
\FOR{$C_i$ in $MCD(C)$}
\STATE Let $C_l$ and $C_r$ be the cluster splits of $C_i$.
\IF{$\Stop_{local}(C_l)$ and $\Stop_{local}(C_r)$ are false}
\STATE $S = S \cup C_i$.
\ELSE \STATE $S = S \cup Split(C_i)$
\ENDIF
\ENDFOR
\STATE Return $S$.
\end{algorithmic}
\label{algo:split}
\end{algorithm}

\paragraph{Computing MCD} We give in Algorithm~\ref{algo:mcd}, a method to find the
minimal core decomposition of a cluster. It checks whether the input
cluster is a core. Otherwise it computes cluster splits $C_l$, and
$C_r$ and computes their MCD recursively. Note that this algorithm is
centralized and takes as input a single cluster, which we assume fits
in the memory of a single machine (unlike connected components, graph
clusters are rather small). 

\paragraph{Stopping Algorithm} Our stopping algorithm, shown in Algorithm~\ref{algo:stop}, is run
after each map-reduce iteration of
Algorithm~\ref{algo:distributed:clustering}. It takes as input the
clustering $\C$ obtained after the map-reduce iteration of
Algorithm~\ref{algo:distributed:clustering} . It starts by computing
the minimal core decomposition, $MCD(C)$, of each cluster $C$ in
$\C$.  This computation can be performed during the reduce step of the pervious map-reduce iteration of
Algorithm~\ref{algo:distributed:clustering}. Then, it runs a
map-reduce iteration of its own. In the map step, using \Ha, each core
$C_i$ is hashed to all machines $m(u)$ for $u \in C_i$. In reducer, for
machine $m(v)$, we pick the incoming core with largest size, say $C_v$. Since $\Stop$ is local,
there exists a local function $\Stop_{local}$. We compute
$\Stop_{local}(C_v)$ to determine whether to stop processing this core
further. Finally, the algorithm stops if all the cores for nodes $v$
in the graph are stopped.

\paragraph{Splitting Clusters} If the stopping algorithm~(Algorithm~\ref{algo:stop}) returns true,
then clustering is complete. However, some clusters could have merged
incorrectly in the final map-reduce iteration done before the stopping
condition was checked. Our recursive splitting algorithm, Algorithm~\ref{algo:split}, correctly splits such a cluster $C$ by first
computing the minimal core decomposition, $MCD(C)$. Then it checks
for each core $C_i \in MCD(C)$ that its cluster splits $C_l$ and $C_r$
could have been merged by ensuring that both $\Stop_{local}(C_l)$
and $\Stop_{local}(C_r)$ are false. If that is the case, then core $C_i$ is valid
and added to the output, otherwise the clusters $C_l$ and $C_r$ should
not have been merged, and $C_i$ is split further.
 
\eat{
\begin{algorithm}[t]
\caption{Stopping Algorithm}
\begin{algorithmic}[1]
\STATE {\bf Input:} Stopping predicate $\Stop$, Cluster $C \subseteq V$.
\STATE {\bf Output:} A set $V$ of machines to stop 
\IF{$\Stop_{local}(C)$ is false}
 \STATE $V = \{\}$
\ELSE
 \STATE Compute minimal core decomposition, $MCD(C)$, of $C$.
\FOR{Core $C_i$ in $MCD(C)$}
 \IF{Closest cluster $C_j$ to $C_i$ also in MCD(C)}
 \STATE $V = V \cup C_i$
 \ENDIF
\ENDFOR
\ENDIF
\STATE Send $stop$ signal to all machines in $V$
\end{algorithmic}
\label{algo:stop}
\end{algorithm}
}

\subsection{Correctness \& Complexity Results}

We first show the correctness of Algorithm~\ref{algo:distributed:clustering}. For that we first show the following lemma about the validity of cores.

\begin{lemma}[Cores are valid] \label{lemma:cores} Let $\C_{central}$ be the output of Algorithm~\ref{algo:central}, and $C$ be any core (defined according to Def.~\ref{def:core})  such that its clusters splits $C_l$, $C_r$ have both $\Stop_{local}(C_l)$ and $\Stop_{local}(C_r)$ as false. Then $C$ is valid, i.e. Algorithm~\ref{algo:central} does compute $C$ some time during its execution, and there exists a cluster $C_{central}$ in $\C_{central}$ such that $C \subseteq C_{central}$.
\end{lemma}
\begin{proof} The proof uses induction. For the base case, note that any singleton core is obviously valid. Now assume that $C$ has cluster splits $C_l$ and $C_r$, which  by induction hypothesis, are valid. Then we show that $C$ is also valid. Since $\Stop_{local}(C_l)$ and $\Stop_{local}(C_r)$ are false for the cluster splits of $C$, they do get merged with some clusters in Algorithm~\ref{algo:central}. Furthermore, by definition of a core, $C_l, C_r$ are closest to each other, hence they actually get merged with each other.Thus $C = C_l \cup C_r$ is constructed some during execution of Algorithm~\ref{algo:central}, and there exists a cluster $C_{central}$ in its output that  contains $C_l \cup C_r = C$, completing the proof.
\end{proof}

Next we show the correctness of Algorithm~\ref{algo:distributed:clustering}. Due to lack of space, the proof is omitted and appears in the Appendix. 

\begin{theorem}[Correctness] \label{thm:clustering:correctness} The distributed Algorithm~\ref{algo:distributed:clustering} simulates the centralized Algorithm~\ref{algo:central}, i.e., it outputs the same clustering as Algorithm~\ref{algo:central}. 
\end{theorem}

Next we state the complexity result for \slc. We omit the proof as it is very similar to that of the complexity result for connected components.
\begin{theorem}[Single-linkage Runtime] \label{thm:single:linkage}  If \Ha is used in Algorithm~\ref{algo:distributed:clustering}, then it finishes in $O(\log{n)}$ map-reduce iterations and $O(n\nnodes + \nedges)$ communication per iteration, where $n$ denotes the size of the largest cluster.
\end{theorem}

We also conjecture that if \Hm is used in Algorithm~\ref{algo:distributed:clustering}, then it finishes in $O(\log{d})$ steps.

\section{Experiments}
\label{sec:experiments}
In this section, we experimentally analyze the performance of the proposed algorithms for computing connected components of a graph. We also evaluate the performance of our agglomerative clustering algorithms.

\newcommand{\synth}{{\sf {\em Synth}}\xspace}
\newcommand{\movie}{{\sf {\em Movie}}\xspace}
\newcommand{\moview}{{\sf {\em MovieW}}\xspace}
\newcommand{\biz}{{\sf {\em Biz}}\xspace}
\newcommand{\bizw}{{\sf {\em BizW}}\xspace}
\newcommand{\yahoo}{{\sf {\em Social}}\xspace}
\newcommand{\yahoos}{{\sf {\em SocialSparse}}\xspace}
\newcommand{\twitter}{{\sf {\em Twitter}}\xspace}
\newcommand{\twitters}{{\sf {\em TwitterSparse}}\xspace}

\begin{figure*}[t]
\centering
\subfigure[\# of Iterations (max over 10 runs)]{
\input{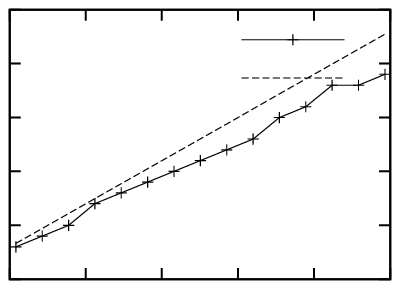}
\label{graph:pathHM_iter}
}
\subfigure[Largest intermediate data size (average 10 runs)]{
\input{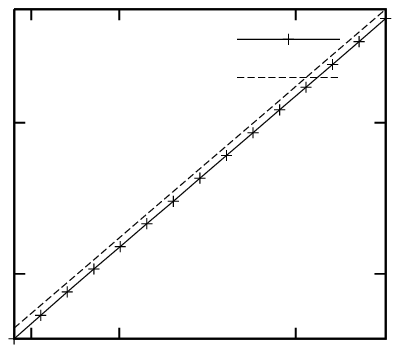}
\label{graph:pathHM_mem}
}
\caption{Analysis of \Hm on a path graph with random node orderings}
\end{figure*}

\begin{figure*}[t]
\centering
\subfigure[\# of Iterations (max over 10 runs)]{
\input{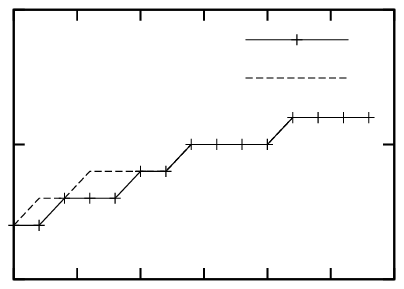}
\label{graph:treeHM_iter}
}
\subfigure[Largest intermediate data size (average 10 runs)]{
\input{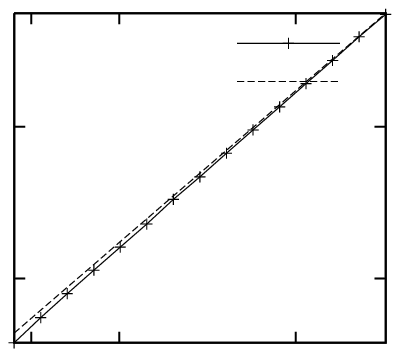}
\label{graph:treeHM_mem}
}
\caption{Analysis of \Hm on a tree graph with random node orderings}
\end{figure*}

\paragraph{Datasets:} To illustrate the properties of our algorithms we use both synthetic and real datasets.
\squishlist
\item \movie: The movie dataset has movie listings collected from Y! Movies\footnote{\url{http://movies.yahoo.com/movie/*/info}} and DBpedia\footnote{\url{http://dbpedia.org/}}. Edges between two listings  correspond to movies that are duplicates of one another; these edges are output by a pairwise matcher algorithm. Listings maybe duplicates from the same or different source, hence the sizes of connected components vary. The number of nodes $\nnodes = 431,221$~(nearly 430K) and the number of edges $\nedges = 889,205$~(nearly 890K). We also have a sample of this graph (238K nodes and 459K edges) with weighted edges, denoted \moview, which we use for agglomerative clustering experiments.

\item \biz: The biz dataset has business listings coming from two overlapping feeds that are licensed by a large internet company. Again edges between two businesses correspond to business listings that are duplicates of one another. Here, $\nnodes = 10,802,777$ (nearly 10.8M) and $\nedges = 10,197,043$ (nearly 10.2M). We also have a version of this graph with weighted edges, denoted \bizw, which we use for agglomerative clustering experiments.

\item \yahoo: The \yahoo dataset has social network edges between users of a large internet company. \yahoo has $\nnodes = 58552777$ (nearly 58M) and $\nedges = 156355406$ (nearly 156M). Since social network graphs have low diameter, we remove a random sample of edges, and generate \yahoos.  With \nedges = $15,638,853$ (nearly 15M), \yahoos graph is more sparse, but has much higher diameter than \yahoo.

\item \twitter: The \twitter dataset (collected by Cha et al \cite{mpi-twitter}) has follower relationship between twitter users. \twitter has $\nnodes = 42069704$ (nearly 42M) and $\nedges = 1423194279$ (nearly 1423M). Again we remove a random sample of edges, and generate a more sparse graph,  \twitters, with \nedges = $142308452$ (nearly 142M).

\item \synth: We also synthetically generate graphs of a varying diameter and sizes in order to better understand the properties of the algorithms.
\squishend

\subsection{Connected Components}
\paragraph{Algorithms:} We compare \Hma, \Hcm, \Ha, \Hm and its load-balanced version $\mbox{\Hm}^*$ (Section~\ref{sec:implementation}). For \Hma, we use the open-source Pegasus implementation\footnote{\url{http://www.cs.cmu.edu/~pegasus/}}, which has several optimizations over the \Hma algorithm. We implemented all other algorithms in Pig\footnote{\url{http://pig.apache.org/}} on Hadoop\footnote{\url{http://hadoop.apache.org}}. There is no native support for iterative computation on Pig or Hadoop map-reduce. 
We implement one iteration of our algorithm in Pig and drive  a loop using a python script. Implementing the algorithms on iterative map-reduce platforms like HaLoop~\cite{bu10:haloop} and Twister~\cite{ekanayake10:twister} is an interesting avenue for future work.

\begin{figure*}
\vspace{-1.5in}
\centering
\subfigure[Group I:Runtimes (in minutes)]{
\includegraphics[width=0.8\columnwidth]{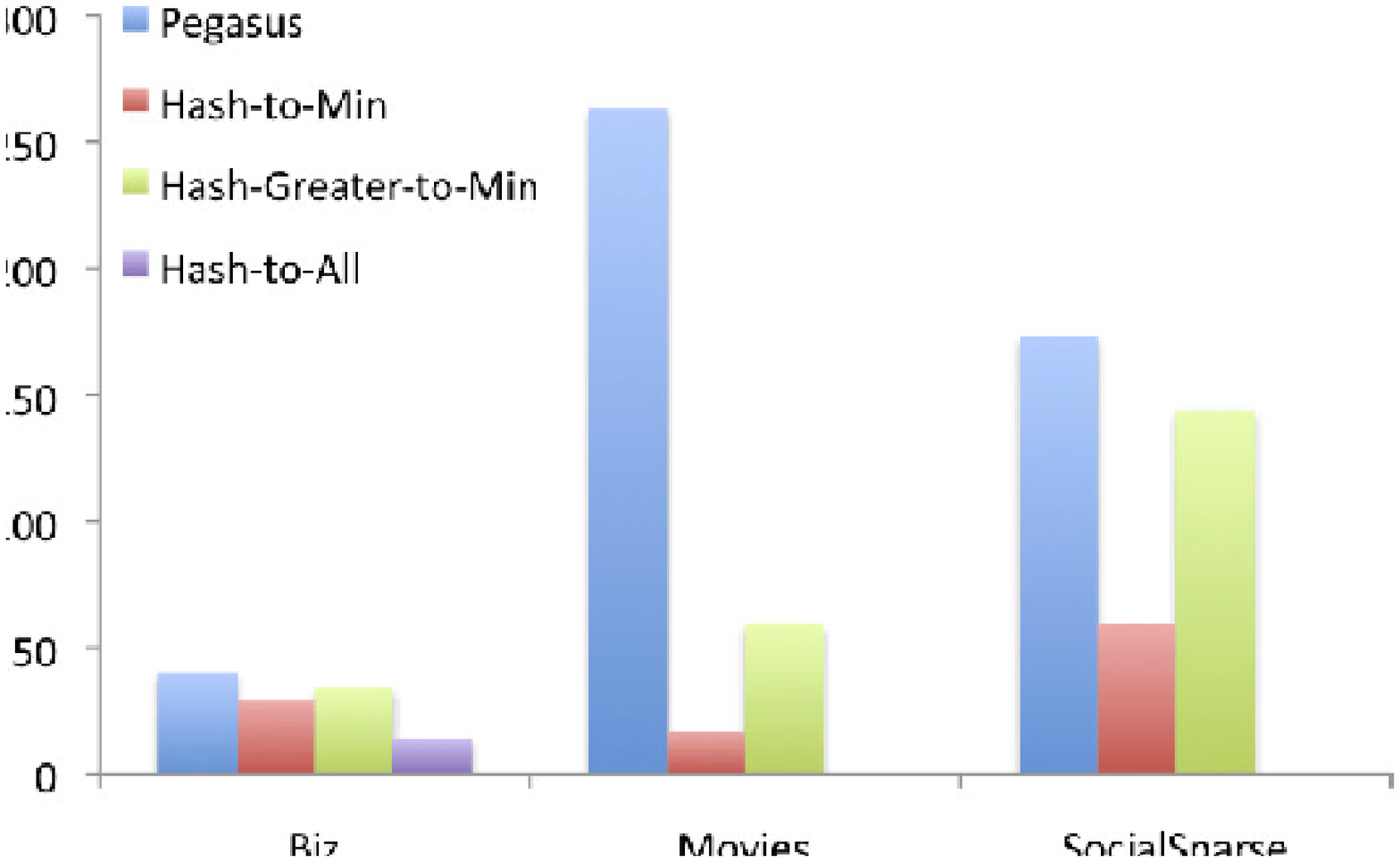}
\label{fig:cc:bars_time}
}
\subfigure[Group I: \# of Map-Reduce jobs]{
\includegraphics[width=0.8\columnwidth]{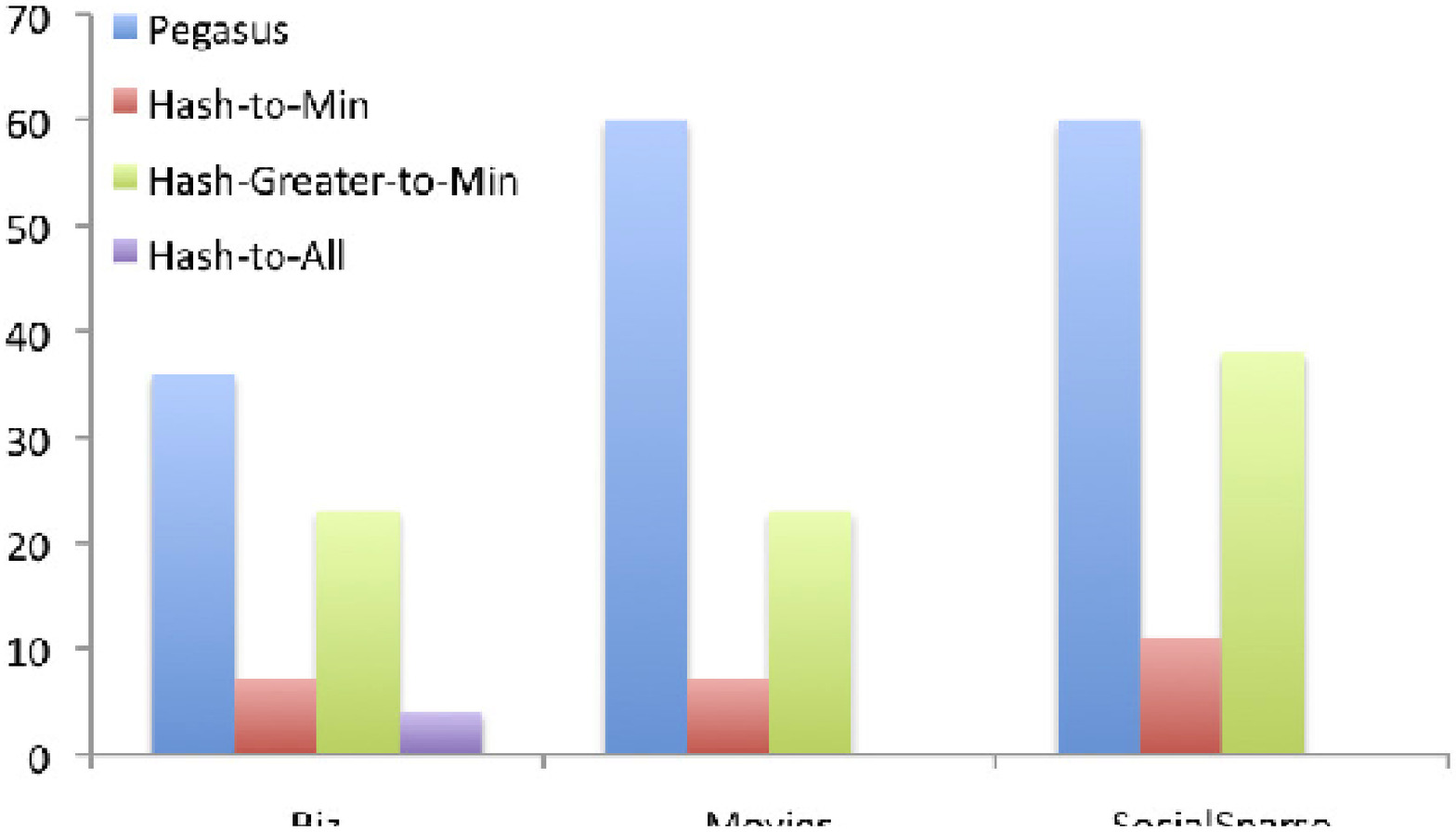}
\label{fig:cc:bars_rounds}
}
\subfigure[Group II: Runtimes (in minutes)]{
\includegraphics[width=0.8\columnwidth]{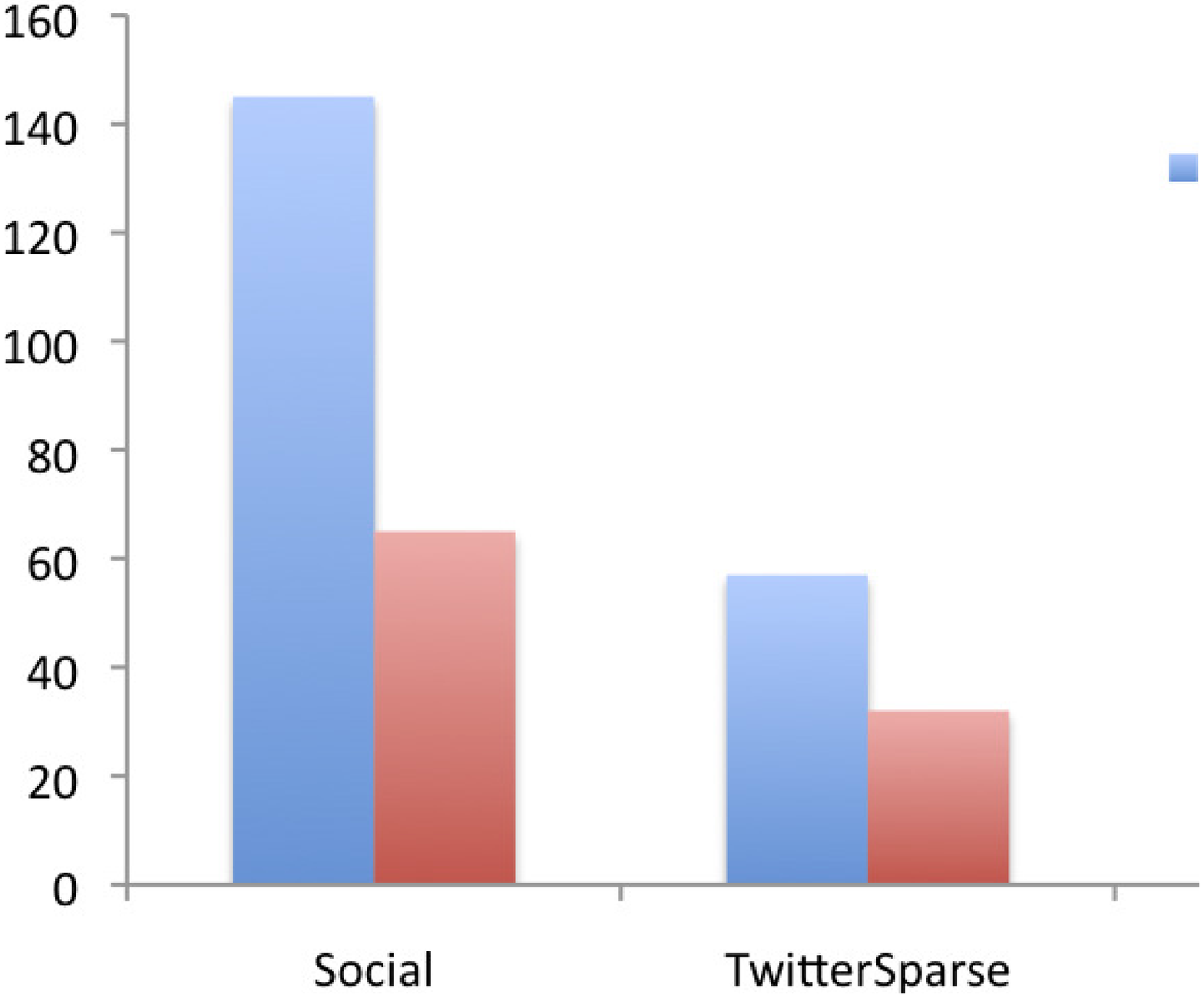}
\label{fig:cc:bars_time:group2}
}
\subfigure[Group II: \# of Map-Reduce jobs]{
\includegraphics[width=0.8\columnwidth]{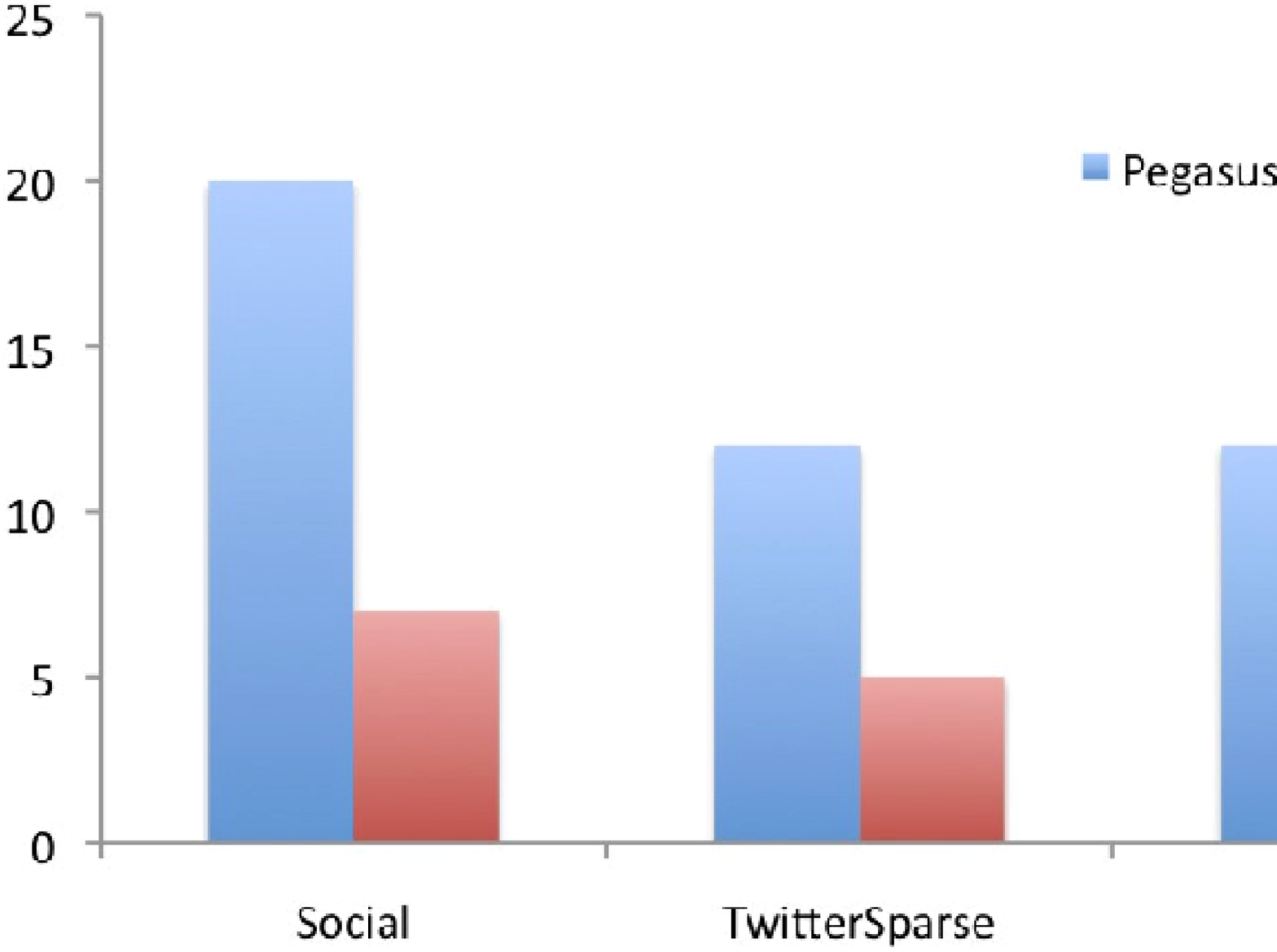}
\label{fig:cc:bars_rounds:group2}
}
\caption{Comparison of Pegasus and our algorithms on real datasets.}
\end{figure*}

\begin{table*}[t]
\centering
\begin{tabular}{ |c|c|c|c|c|c|c|c|c|c|c|c|c| }
\hline
Input & $\nnodes$  & $\nedges$ & $n$ & \multicolumn{2}{|c|}{Pegasus} &  \multicolumn{2}{|c|}{\Hm} &   \multicolumn{2}{|c|}{\Hcm}  &\multicolumn{2}{|c|}{\Ha} \\
\hline
&&&& \# MR jobs & Time &  \# MR jobs & Time  & \# MR jobs & Time & \# MR jobs & Time\\
\hline
\biz & 10.8M & 10.1M &  93 & 36 & 40 & 7 &  29 & 23 & 34 & 4 & {\bf 14} \\
\movie & 430K & 890K & 17K & 60 & 263 & 7 & {\bf 17} & 23 & 59 & DNF & DNF \\
\yahoos & 58M  & 15M & 2.9M  & 60 & 173 & 11 & {\bf 59} & 38 & 144 & DNF & DNF \\
 \hline
\end{tabular}
\caption{Comparison of Pegasus, \Hm,  \Hcm, and \Ha on the Group I datasets. Time is averaged over 4 runs and rounded to minutes. Optimal times appear in bold: in all cases either \Hm or \Ha is optimal.}
\label{table:real:all}
\end{table*}

\begin{table*}[t]
\centering
\begin{tabular}{ |c|c|c|c|c|c|c|c| }
\hline
Input & $\nnodes$  & $\nedges$ & $n$ & \multicolumn{2}{|c|}{Pegasus} &  \multicolumn{2}{|c|}{$\mbox{\Hm}^*$}  \\
\hline
&&&& \# MR jobs & Time &  \# MR jobs & Time \\
\hline
\yahoo & 58M  & 156M & 36M & 20 & 145 &  7 &  {\bf 65}  \\
\twitters & 42M  & 142M & 24M & 12 &  57 &  5 &  {\bf 32} \\
\twitter & 42M  & 1423M & 42M & 12 &  61 &  5 &  {\bf 50} \\
 \hline
\end{tabular}
\caption{Comparison of Pegasus and the $\mbox{\Hm}^*$ algorithm on Group II datasets. Time is averaged over 4 runs and rounded to minutes. Optimal times appear in bold: in all cases, $\mbox{Hm}^*$ is optimal.}
\label{table:real:group2}
\end{table*}

\begin{figure*}[t]
\vspace{-1.2in}
\centering
\subfigure[Runtimes (in minutes)]{
\includegraphics[width=0.8\columnwidth]{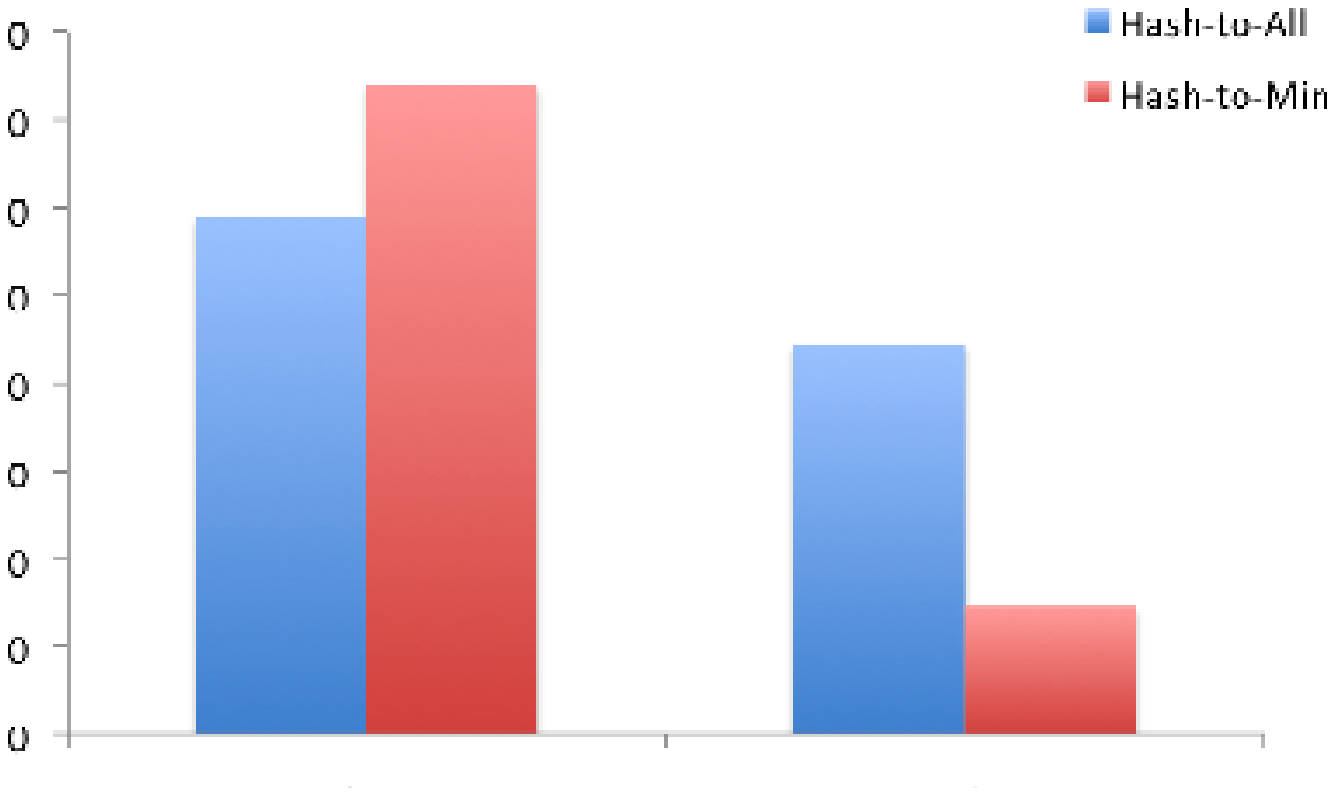}
\label{fig:slc:bars_time}
}
\subfigure[\# of Map-Reduce jobs]{
\includegraphics[width=0.8\columnwidth]{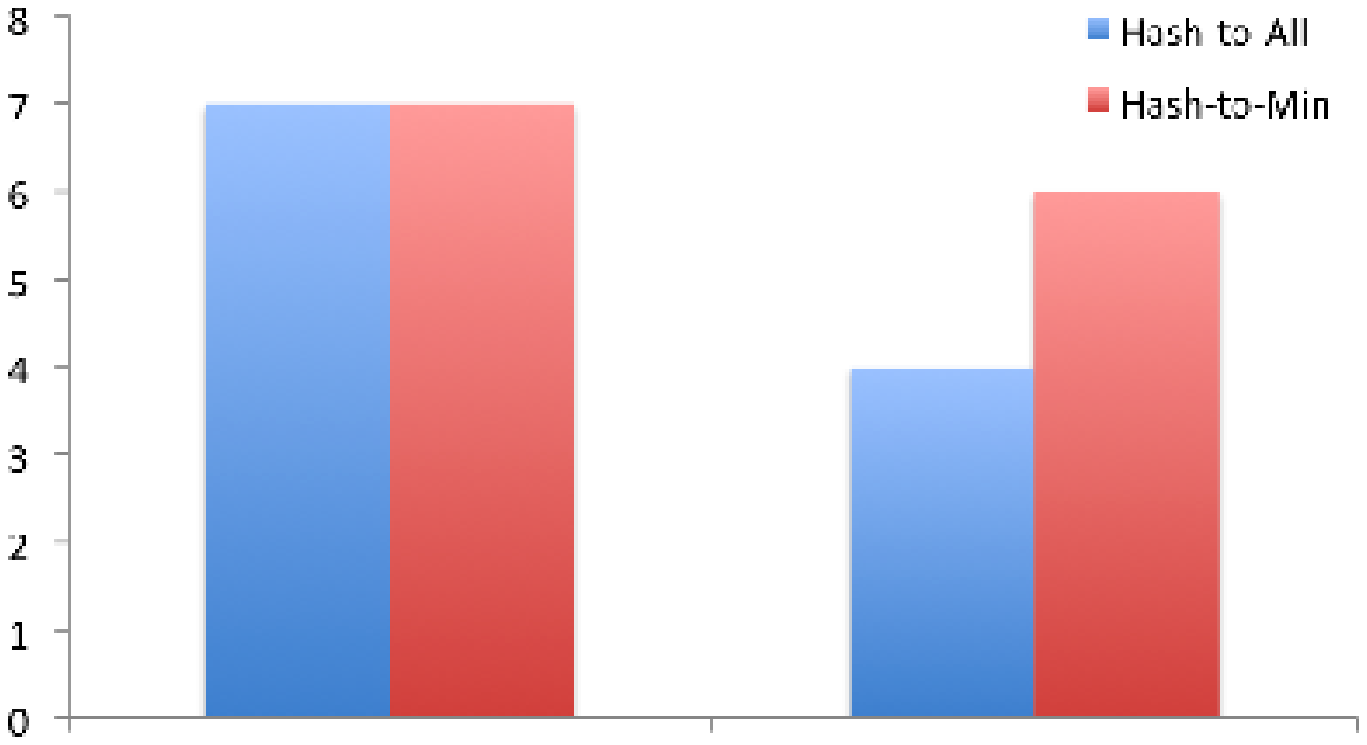}
\label{fig:slc:bars_rounds}
}
\vspace{-0.1in}
\caption{Comparison of \Ha and  \Hm for single linkage clustering on \bizw and \moview. }
\end{figure*}

 \subsubsection{Analyzing \Hm on Synthetic Data}
 We start by experimentally analyzing the rounds complexity and space requirements of \Hm. We run it on two kinds of synthetic graphs: paths and complete binary trees. We use synthetic data for this experiment so that we have explicit control over parameters  $d$, $\nnodes$, and $\nedges$. Later we report the performance of \Hm on real data as well. We use path graphs since they have largest  $d$ for a given $\nnodes$ and complete binary trees since they give a very small $d = \log{\nnodes}$. 

For measuring space requirement, we measure the largest intermediate data size in any iteration of \Hm. Since the performance of \Hm depends on the random ordering of the nodes chosen, we choose $10$ different random orderings for each input. For number of iterations, we report the worst-case among runs on all random orderings, while for intermediates data size we average the maximum intermediate data size over all runs of \Hm. This is to verify our conjecture that number of iterations is $2 \log{d}$ in the worst-case (independent of node ordering) and intermediate space complexity is $O(\nnodes + \nedges)$ in expectation (over possible node orderings).

For path graphs, we vary the number of nodes from $32$~($2^5$) to $524,288$~($2^{19})$. In Figure~\ref{graph:pathHM_iter}, we plot the number of iterations (worst-case over 10 runs on random orderings) with respect to $\log{d}$. Since the diameter of a path graph is equal to number of nodes, $d$ varies from $32$ to $524,288$ as well. As conjectured the plot is linear and always lies below the line corresponding to $2\log{d}$. In Figure~\ref{graph:pathHM_mem}, we plot the largest intermediate data size (averaged over 10 runs on random orderings) with respect to $\nnodes + \nedges$. Note that both x-axis and y-axis are in log-scale. Again as conjectured, the plot is linear and always lies below $3(\nnodes+\nedges)$.

For complete binary trees, we again vary the number of nodes from $32$~($2^5$) to $524,288$~($2^{19})$. The main difference from the path case is that for a complete binary tree, diameter is $2\log(\nnodes)$ and hence the diameter varies only from $10$ to $38$. Again in Figure~\ref{graph:treeHM_iter}, we see that the rounds complexity still lies below the curve for $2\log{d}$ supporting our conjecture even for trees. In Figure~\ref{graph:treeHM_mem}, we again see that space complexity grows linearly and is bounded by $3(\nnodes + \nedges)$.

\eat{
\begin{table*}[t]
\centering
\begin{tabular}{ |c|c|c|c|c|c|c|c|c|c|c|c|c| }
\hline
Input & $\nnodes$  & $\nedges$ & $n$ & \multicolumn{2}{|c|}{Pegasus} &  \multicolumn{2}{|c|}{\Hm} &   \multicolumn{2}{|c|}{\Hcm}  &\multicolumn{2}{|c|}{\Ha} \\
\hline
&&&& \# of MR jobs & Total Time &  \# of MR jobs & Total Time  & \# of MR jobs & Total Time & \# of MR jobs & Total Time\\
\hline
\biz & 10.8M & 10.1M &  & 36 & 0:40:07 & 8 &  00:29:20 & 4 & 0:14:26 \\
\movie & 430K & 890K &  17K & 60 & 4:23:08 & 15 & 1:14:28 & DNF & DNF \\
\yahoo & 17M  & 15M & 2.9M  & 60 & 2:53:48 & 11 & 00:59:14 & DNF & DNF \\
 \hline
\end{tabular}
\caption{Comparison of Pegasus, \Hm, and \Ha on real datasets. In all cases either \Hm or \Ha is optimal in terms of total time.}
\end{table*}
}

\subsubsection{Analysis on Real Data}
We next compared \Hm, \Hcm, and \Ha algorithms on real datasets against Pegasus~\cite{pegasus}. To the best of our knowledge, Pegasus is the fastest technique on MapReduce for computing connected components.
Although all datasets are sparse (have average degree less than 3), each dataset has very different distribution on the size $n$ of the largest connected components and graph diameter $d$. We partition our datasets into two groups --  {\em group I} with $d>=20$ and relatively small $n$, and {\em group II} with $d<20$ and very large $n$.

\paragraph{Group I: Graphs with large $d$ and small $n$:}
This group includes \biz, \movie, and \yahoos datasets that have large diameters ranging from $20$ to $80$. On account of large diameters, these graphs requires more MR jobs and hence longer time, even though they are somewhat smaller than the graphs in the other group. These graphs have small connected components that fit in memory.

Each connected component in the \biz dataset represents the number of duplicates of a real-world entity. Since there are only two feeds creating this dataset, and each of the two feeds is almost void of duplicates, the size of most connected components is 2. In some extreme cases, there are more duplicates, and the largest connected component we saw had size 93. The \movie dataset has more number of sources, and consequently significantly more number of duplicates. Hence the size of some of the connected components for it is significantly larger, with the largest containing 17,213 nodes. Finally, the \yahoos dataset has the largest connected component in this group, with the largest having 2,945,644 nodes. Table~\ref{table:real:all} summarizes the input graph parameters. It also includes the number of map-reduce jobs and the total runtime for all of the four techniques.

Differences in the connected component sizes has a very interesting effect on the run-times of the algorithm as shown in Figures~\ref{fig:cc:bars_time} and \ref{fig:cc:bars_rounds}. On account of the extremely small size of connected components, runtime for all algorithms is fastest for the \biz dataset, even though the number of nodes and edges in \biz is larger than the \movie dataset. \Ha has the best performance for this dataset, almost 3 times faster than Pegasus. This is to be expected as \Ha just takes 4 iterations~(in general it takes $\log{d}$ iterations) to converge. Since connected components are small, the replication of components, and the large intermediate data size does not affect its performance that much. We believe that \Ha is the fastest algorithm whenever the intermediate data size is not a bottleneck. \Hm takes twice as many iterations ($2\log{d}$ in general) and hence takes almost twice the time. Finally, Pegasus takes even more time because of a larger number of iterations, and a larger number of map-reduce jobs.

For the \movie and \yahoos datasets, connected components are much larger. Hence \Ha does not finish on this dataset due to large intermediate data sizes. However, \Hm beats Pegasus by a factor of nearly $3$ in the \yahoos dataset since it requires a fewer number of iterations. On movies, the difference is the most stark: \Hm has 15 times faster runtime than Pegasus again due to significant difference in the number of iterations.

\paragraph{Group II: Graphs with small $d$ and large $n$:}
This group includes \yahoo, \twitters, and \twitter dataset that have a small diameter of less than $20$, and results are shown in Figures~\ref{fig:cc:bars_time:group2} and~\ref{fig:cc:bars_rounds:group2} and Table~\ref{table:real:group2}. Unlike Group I, these datasets have very large connected components, such that even a single connected component does not fit into memory of a single mapper. Thus we apply our robust implementation of \Hm (denoted $\mbox{\Hm}^*$) described in Sec.~\ref{sec:implementation}. 

The $\mbox{\Hm}^*$ algorithm is nearly twice as fast as pegasus, owing to reduction in the number of MR rounds. Only exception is the \twitter graph, where reduction in times is only 18\%. This is because the \twitter graph has some nodes with very high degree, which makes load-balancing a problem for all algorithms. 


\subsection{\Slc}
We implemented \slc on map-reduce using both \Ha and \Hm hashing strategies. We used these algorithms to cluster the \moview and \bizw datasets. Figures~\ref{fig:slc:bars_time} and \ref{fig:slc:bars_rounds} shows the runtime and number of map-reduce iterations for both these algorithms, respectively. Analogous to our results for connected components, for the \moview dataset, we find that \Hm outperforms \Ha both in terms of total time as well as number of rounds. On the \bizw dataset, we find that both \Hm and \Ha take exactly the same number of rounds. Nevertheless, \Ha takes lesser time to complete that \Hm. This is because some clusters (with small $n$) finish much earlier in \Ha; finished clusters reduce the amount of communication required in further iterations.

\newcommand{\mr}[1]{{\sc mr}-{#1}}
\newcommand{\bsp}[1]{{\sc bsp}-{#1}}

\begin{table}[t]
{\small
\begin{tabular}{|c||c|c||c|c|}
\hline
{\bf k} & \multicolumn{2}{|c||}{{\bf \mr{k} Completion Time}} & \multicolumn{2}{|c|}{{\bf \bsp{k} Completion Time}} \\
& {\bf maximum} & {\bf median} & {\bf maximum} & {\bf median} \\
\hline
1 &  10:45  & 10:45 & 7:40 & 7:40 \\
5 &  17:03  & 11:30 & 11:20 & 8:05 \\
10 &  19:10  & 12:46 & 23:39 & 15:49 \\
15 &  28:17 & 26:07 & 64:51 & 43:37 \\
\hline
\end{tabular}
}
\caption{\label{table:bsp} Median and maximum completion times (in min:sec) for $k$ connected component jobs deployed simultaneously using map-reduce (\mr{k}) and Giraph (\bsp{k})}
\vspace{-0.3in}
\end{table}
\subsection{Comparison to Bulk Synchronous Parallel Algorithms}\label{sec:bspexp}

Bulk synchronous parallel (BSP) paradigm is generally considered more efficient for graph processing than map-reduce as it has less setup and overhead costs for each new iteration. While the algorithmic improvements of reducing number of iterations presented in this paper are important independent of the underlying system used, these improvements are of less significance in BSP due to low overhead of additional iterations.

In this section, we show that BSP does not necessarily dominate Map-Reduce for large-scale graph processing~(and thus our algorithmic improvements for Map-Reduce are still relevant and important). We show this by running an interesting experiment in shared grids having congested environments. We took the \movie graph and computed connected components using \Hm (map-reduce, with 50 reducers) and using \Hma \footnote{\Hma is used as it is easier to implement and not much different than \Hm in terms of runtime for BSP environment} on Giraph \cite{Ching2010:Giraph} (BSP, with 100 mappers), an open source implementation of Pregel \cite{Malewicz2010:Pregel} for Hadoop. We deployed $k = 1$, $5$, $10$, and $15$ copies of each algorithm (denoted by \mr{k} and \bsp{k}), and tracked the maximum and median completion times of the jobs. The jobs were deployed on a shared Hadoop cluster with 454 map slots and 151 reduce slots, and the cluster experienced normal and equal load from other unrelated tasks.

Table~\ref{table:bsp} summarizes our results. As expected, \bsp{1} outperforms \mr{1};\footnote{The time taken for \mr{1} is different in Tables~\ref{table:real:all} and \ref{table:bsp} since they were run on different clusters with different number of reducers (100 and 50 resp.).} unlike map-reduce, the BSP paradigm does not have the per-iteration overheads. However, as $k$ increases from 1 to 15, we can see that the maximum and median completion times for jobs increases at a faster rate for \bsp{k} than for \mr{k}. This is because the BSP implementation needs to hold all 100 mappers till the job completes. On the other hand, the map-reduce implementation can 
naturally parallelize the map and reduce rounds of different jobs, thus eliminating the impact of per round overheads.  So it is not surprising that while all jobs in \mr{15} completed in about 20 minutes, it took an hour for jobs in \bsp{15} to complete. Note that the cluster configurations favor BSP implementations since the reducer capacity (which limits the map-reduce implementation of \Hm) is much smaller ($< 34\%$) than the mapper capacity (which limits the BSP implementation of \Hma). We also ran the experiments on clusters with higher ratios of reducers to mappers, and we observe similar results (not included due to space constraints) showing that map-reduce handles congestion better than BSP implementations.

\eat{
\ashwinnote{
\begin{itemize}
\item {\em Number of rounds (Synth):} On synthetic data show the number of rounds versus $d$ or $\log d$.
\item {\em Number of rounds (Real):} On real graphs show the number of rounds
\item {\em Communication cost:} On real data show the total communication cost (using Map-reduce logs)
\item {\em Any other ..... ??}
\end{itemize}
}
}

\section{Conclusions and Future Work}
\label{sec:conclusions}
In this paper we considered the problem of find connected components in a large graph. We proposed the first map-reduce algorithms that can find the connected components in logarithmic number of iterations -- {\em (i)} \Hcm, which provably requires at most $3\log{n}$ iterations with high probability, and at most $2(\nnodes + \nedges)$ communication per iteration, and {\em (ii)} \Hm, which has a worse theoretical complexity, but  in practice completes in at most $2 \log d$ iterations and $3(\nnodes + \nedges)$ communication per iteration; $n$ is the size of the largest component and $d$ is the diameter of the graph. 
We showed how to extend our techniques to the problem of single linkage clustering, and proposed the first algorithm that computes a clustering in provably $O(\log n)$ iterations. 

\balance
\bibliographystyle{abbrv}
\bibliography{bib/disconnect}

\normalsize{
\clearpage
\onecolumn
\appendix
\section{Connected Components}
\subsection{Proof of Theorem~\ref{thm:Hm:rounds}}
\label{sec:Hm}

We first restate Theorem~\ref{thm:Hm:rounds} below.

\begin{theorem}[\ref{thm:Hm:rounds}]
 Let $G=(V,E)$ be a path graph (i.e. a tree with only nodes of degree 2 or 1). Then, \Hm correctly computes the connected component of  $G = (V,E)$ in $4\log{n}$ map-reduce rounds.
\end{theorem}

To prove the above theorem, we first consider path graphs when node ids increase from left to right. Then we show the result for path graphs with arbitrary ordering.

\begin{lemma}\label{lemma:path}
Consider a path with node ids increasing from the left to right. Then after $k$ iterations of the Hash-to-min algorithm,
\begin{itemize}
\item For every node $j$ within a distance of $2^k$ from the minimum node $m$, $m$ knows $j$ and $j$ knows $m$.
\item For every pair of nodes $i, j$ that are a distance $2^k$ apart, $i$ knows $j$ and $j$ knows $k$.
\item Node $j$ is not known to and does not know any node $i$ that is at a distance $>2^k$.
\end{itemize}
\end{lemma}
\begin{proof}
The proof is by induction. \\
{\em Base Case:} After 1 iterations, each node knows about its 1-hop and 2-hop neighbors (on either side).  \\
{\em Induction Hypothesis:} Suppose the claim holds after $k-1$ iterations. \\
{\em Induction Step:} \\
In the $k^{th}$ iteration, consider a node $j$ that is at a distance $d$ from the min node $m$, where $2^{k-1} < d \leq  2^k$. From the induction hypothesis, there is some node $i$ that is $2^{k-1}$ away from $j$ that knows $j$. Since, $m$ is known to $i$ (from induction hypothesis), the Hash-to-min algorithm would send $j$ to $m$ and $m$ to $j$ in the current iteration. Therefore, $m$ knows $j$ and $m$ is known to $j$.

Consider a node $j$ that is $>2^k$ distance from the min node. At the end of the previous iteration, $j$ knew (and was known to) $i$, and $i$ knew and was known to $i'$ -- where $i$ and $i'$ are at distance $2^{k-1}$ from $j$ and $i$ respectively. Moreover, $i$ did not know any node $i''$ smaller than $i'$. Therefore, in the current iteration, $i$ sends $i'$ to $j$ and $j$ to $i'$. Therefore, $j$ knows and is known to a node that is $2^k$ distance away.

Finally, we can show that a node does not know (and is not known to) any node that is distance $>2^k$ as follows. Node $j$ can only get a smaller node $i'$ if $i'$ is a minimum at some node $i$. Since in the previous step no one knows a node at distance $>2^{k-1}$, $j$ cannot know a node at distance $>2^k$. 
\end{proof}

Now we extend the proof for arbitrary path graphs. Denote $min_k(u)$ the minimum node after $k$ iterations that $u$ knows that also knows $u$. Also denote $\Delta(u,v)$ the distance between node $u$ and $v$.
\begin{definition}[Local Minima]
A node $v$ is local minimum if all its neighbors have id larger than $v$'s id.
\end{definition}

For a path graph, we define the notion of levels below.
\begin{definition}[Levels] Given a path, level $0$ consists of all nodes in the path. Level $i$ is then defined recursively as nodes that are local minimum nodes among the nodes at level $i-1$, if the level $i-1$ nodes are arranged in the order in which they occur in the path. Denote the set of nodes at level $i$ as $level(i)$.
\end{definition}

\begin{proposition} \label{prop:levels:number} The number of levels having more than $1$ node is at most $\log{n}$.
\end{proposition}
\begin{proof} The proof follows from the fact that for each level $i$, no consecutive nodes can be local minimum. Hence $|level(i)| \leq |level(i-1)/2)$.
\end{proof}

\begin{lemma}\label{lemma:path:levels}
Consider a path $P$ with three segments $P_1, P_2$ and $P_3$, where $P_1$ and $P_3$ are arbitrary, and $P_2$ has $r+1$  level $\ell$ nodes $l_1,l_2,...l_r,m_1$  going from left to right. Assume that labels are such that $l_1 < l_2 < \ldots < l_r$.  For a node $l_i$, denote $\lr(l_i,k,\ell)$ the closest level $\ell$ node $l_j$ from $l_i$ towards the right such that $min_k(l_j) > l_i$. Denote $T(k,\ell)$ and $M(k,\ell)$ as $$T(k,\ell) = min_{P,l_i: \lr(l_i,k,\ell) \neq l_r} \Delta(l_i,\lr(l_i,k,\ell))$$ and $$M(k,\ell) = min_{P: min_k(l_1) > min_{k-1}(m_1)\mbox{ or }min_k(m_1) > min_{k-1}(l_1)} \Delta(l_1,m_1)$$.

Then the following is true:
\begin{enumerate}


\item $T$ satisfies the following recurrence realtion
            \begin{eqnarray*} 
              T(k,\ell) \geq min(T(k-2,\ell) + min(T(k-1,\ell), M(k,\ell), M(k-1,\ell-1) )
             \end{eqnarray*} 
             
\item $M$ satisfies the following recurrence relation
        \begin{eqnarray*}
           M(k,\ell) \geq  min(T(k-2,\ell),M(k-1,\ell-1))
        \end{eqnarray*}
        
\end{enumerate}
\end{lemma}

\begin{proof}
Denote by $[l_t,l_u]$ the level  $\ell$ nodes between $l_t$ and $l_u$, and $[l_t,l_u)$ those nodes except for $l_u$.


{\bf Proof of claim 1:} Let $l_s = \lr(l_i,k-1,\ell)$ and let $l_t$ be the level $\ell$ node just to the left of $l_s$. We know by definition of $l_s=\lr(l_i,k-1,\ell)$ that  $min_{k-1}(l_s) >  l_i$, but for all $l \in [l_i,l_t], min_{k-1}(l) \leq l_i$. Now there are three cases:

\begin{enumerate}
\item $min_{k-1}(l_t) \in P_3$. Then $min_{k-1}(l_t) \leq l_i$~(from above), and any node $l \in [l_t,l_r]$ would have $min_{k-1}(l) \in P_3 \leq min_{k-1}(l_t)$. Hence $\min_{k-1}(l) \leq l_i$ for all $l \in \{l_i,\ldots,l_r\}$. Thus $\lr(l_i,k-1) = l_r$.

\item $min_{k-1}(l_t) \notin P_3$ and $\Delta(l_i,l_t) \geq T(k-1,\ell)$. Denote $l_u =  \lr(l_t,k-1,\ell)$.  Consider any $l \in [l_t,l_u)$. Let $a=min_{k-1}(l)$. Then from the definition of $l_u =  \lr(l_t,k-1)$ and since $l \in [l_t,l_u)$, we know that, $a = min_{k-1}(l) \leq l_t$. Now there are two sub-cases: 

\begin{enumerate} 
\item For all $l \in [l_t,l_u)$, $a = min_{k-1}(l)  \notin P_3 \cup \{m_1\} $. In this case, we show that $\lr(l_i,k,\ell)$ is a level $\ell$ node in $[l_u,l_r]$. For this we will show that for all $l \in [l_t,l_u)$ and $a = min_{k-1}(l)$,  $min_{k-1}(a) \leq l_i$.  If $a \notin P_3 \cup \{m_1\}$, then either (i) $a \in P_1$, in which case  $a \leq l_1$~(otherwise $min_{k-1}(l)$ would be $l_1$ and not $a$), and hence obviously $min_{k-1}(a) \leq l_i$, or (ii) $a \in P_2$ but $a \neq m_1$, and then since $a \leq l_t$~(from above), $a \in [l_i,l_t]$. Hence $min_{k-1}(a) \leq l_i$. 

In other words, after $k-1$ iterations, the minimum for any node $l \in [l_t,l_u)$ is $a$, which in turn has a minimum $b \leq l_i$. Thus in one \Hm step, $b$ would become $min_k(l)$. Hence after $k$ iterations and for any local minimum $l$ between $l_i$ and $l_u$, we have $min_k(l) \leq l_i$.  This shows that $\lr(l_i,k,\ell) \in [l_u,l_r]$. Then either $l_u$ is $l_r$ or the following holds.
 
$$ \Delta(l_i,\lr(l_i,k,\ell))  \geq \Delta(l_i,l_t) + \Delta(l_t,l_u) \geq  T(k-2,\ell) + T(k-1,\ell-1)$$

\item  There exists $l \in [l_t,l_u)$, such that $a = min_{k-1}(l)  \in P_3$. Let $l_w = \lr(l_i,k,\ell)$. Then $l_w$ is the first node in $[l_i,l_r]$ for which $min_{k}(l_w) > l_i$. If $\Delta(l_v,l_w) \leq M(k,\ell)$ then by definition after $k$ iterations $min_k(l_w) \leq min_{k-1}(l_t) \leq l_i$. Thus $\Delta(l_w,l_i) \geq M(k,\ell)$. Thus 
$$ \Delta(l_i,\lr(l_i,k,\ell))  \geq \Delta(l_i,l_t) + \Delta(l_t,l_w) \geq  T(k-2,\ell) + M(k,\ell)$$
\end{enumerate}

\item $min_{k-1}(l_t) \in P_1 || P_2$ and $\Delta(l_i,l_t) \leq T(k-2,\ell)$. In this case, we argue that  $\Delta(l_t,l_s) > M(k-1,\ell-1)$. Assume the contrary: $\Delta(l_t,l_s) \leq M(k-2,\ell-1)$. Since, $\Delta(l_i,l_t) \leq T(k-2,\ell)$, we know that $min_{k-2}(l_t) \leq l_i$. Since $l_t$ and $l_s$ are consecutive level $\ell$ nodes, all the level $\ell-1$ nodes between them are ordered.  Hence by definition of $M(k-1,\ell-1)$, and the fact that $\Delta(l_t,l_s) \leq M(k-1,\ell-1)$, $min_{k-1}(l_s) \leq min_{k-2}(l_t) \leq l_i$ . This contradicts the assumption that $l_s = \lr(l_i,k-1,\ell)$. Thus $\Delta(l_t,l_s) > M(k-1,\ell-1)$.
$$ \Delta(l_i,\lr(l_i,k))  \geq \Delta(l_i,l_t) + \Delta(l_t,l_s) \geq M(k-1,\ell-1) $$
\end{enumerate}

Combining the above cases we complete the proof of claim 1.\\

{\bf Proof of claim 2:}  If $\Delta(l_1,l_r) \leq T(k-2,\ell)$, then $min_{k-2}(l_r) \leq l_1$.  Also if $\Delta(l_r,m_1) \leq M(k-1,\ell-1)$, then $min_{k-1}(m_1) \leq min_{k-2}(l_r) \leq l_1$. If both  $\Delta(l1,l_r) \leq T(k-2,\ell)$ and $\Delta(l_r,m_1) \leq M(k-1,\ell-1)$, then $min_{k-1}(m_1) \leq  l_1$. Hence by claim 1, $min_{k}(m) \leq min_{k-1}(l_1)$. This completes the proof.

\end{proof}

\begin{lemma} \label{lemma:path:cost} Let $T(k,\ell)$ be the quantity as defined in Lemma~\ref{lemma:path:levels}. Then $T(k,\ell) \geq 2^{k/2-\ell}$.
\end{lemma}
\begin{proof} We prove the lemma using induction.

{\em Base Cases: (i) $\ell=0$ and $k\geq1$. Then by  Lemma~\ref{lemma:path}, $T(k,0) \geq 2^{k} \geq 2^{k/2 - 0}$. (ii) $k=1$ and $\ell \geq 1$}. For any level $\ell \geq 1$, $T(1,\ell) \geq 1 \geq 2^{1/2-\ell}$.

{\em Induction Hypothesis (IH)} For all $k_0 \leq k-1$ and $\ell_0 \leq \ell-1$,  $T(k_0,\ell_0) \geq 2^{k_0/2-\ell_0}$

{\em Induction Step:} By Lemma~\ref{lemma:path:levels}, we know that:
\begin{eqnarray*}
  T(k,\ell) &\geq& min \left(T(k-1,\ell) + T(k-2,\ell), T(k-2,\ell-1)\right) \\
                 &\geq& min \left(2^{(k-1)/2-\ell} + 2^{(k-2)-\ell}, 2^{(k-2)/2-\ell+1}\right)~~~~\mbox{(using IH)}\\
                 &\geq& min \left(2^{k/2-\ell}(1/2 + 1/\sqrt{2}), 2^{k/2-\ell}\right) \geq 2^{k/2-\ell}\\
\end{eqnarray*}
\end{proof}

Finally, we can complete the proof of Theorem~\ref{thm:Hm:rounds}.  Since $T(k,l)$ is less than the length of the path, we know that $T(k,l) < n$. Now from Prop.~\ref{prop:levels:number},  the number of levels having more than $1$ node is at most $\log{n}$. Hence $\ell \leq \log{n}$. Finally, from Lemma~\ref{lemma:path:cost}, we know that $T(k,\ell) \geq 2^{k/2-\ell}$. Thus $2^{k/2 - \log{n}} \leq 2^{k/2-\ell} \leq n$. Thus $k \leq 4\log{n}$. This completes the proof of Theorem~\ref{thm:Hm:rounds}.

\subsection{Proof of Theorem~\ref{thm:hgm}}

We first restate Theorem~\ref{thm:hgm} below.

\begin{theorem}[\ref{thm:hgm}]
 Algorithm \Hcm correctly computes the connected components of  $G = (V,E)$ in expected $3\log{n}$ map-reduce rounds (expectation is over the random choices of the node ordering) with  $2(\nnodes + \nedges)$ communication per round in the worst case.
 \end{theorem}
\begin{proof} After $3k$ rounds, denote $M_k =  \{ min(\TC_v) : v \in V \}$ the set of nodes that appear as minimum on some node. For a minimum node $m \in M_k$, denote $\set_k(m)$ the set of all nodes $v$ for which $m = min(\TC_v)$. Then by Lemma~\ref{lemma:hgm}, we know that $\set_k(m) = \core(m)$ after $3k$ rounds. Obviously $\cup_{m \in M_k} \set_k(m) = V$ and for any $m,m' \in M_k$, $\set_k(m) \cap \set_k(m') = \emptyset$.

Consider the graph $G_{M_k}$ with nodes as $M_k$ and an edge between $m \in M_k$ to $m' \in M_k$ if there exists $v \in \set_k(m)$ and $v' \in  \set_k(m')$ such that $v,v'$ are neighbors in the input graph $G$.  If a node $m$ has no outgoing edges in $G_{M_k}$, then $\set_k(m)$ forms a connected component in $G$ disconnected from  other components, this is because, then for all $v' \notin \set_k(m)$,  there exists no edge to $v \in \set_k(m)$.

We can safely ignore such sets $\set_k(m)$. Let $MC_k$ be the set of nodes in $G_{M_k}$ that have at least one outgoing edge. Also if $m \in MC_k$ has an edge to $m' < m$ in $G_{M_k}$, then $m$ will no longer be the minimum of nodes $v \in \set _k(m)$ after $3$ additional rounds. This is because there exist nodes $v \in \set_k(m)$ and $v' \in \set_k(m')$, such that $v$ and $v'$ are neighbors in $G$. Hence in the first round of \Hma, $v'$ will send $m'$ to $v$. In the second round of \Hma, $v$ will send $m'$ to $m$. Hence finally $m$ will get $m'$, and in the round of \Hcm, $m$ will send $\set_k(m)$ to $m'$.

If $|MC_k| = l$, W.L.O.G, we can assume that they are labeled $1$, $2$, \ldots, $l$ (since only relative ordering between them matters anyway).  For any set, $\set_k(m)$, the probability that it its min $m' \in (l/2,l]$ after 3 more rounds is $1/4$. This is because  that happens only when $m \in [(l/2,l)$ and all its neighbors $m' \in G_{M_k}$ are also in $(l/2,l]$. Since there exist at least one neighbor $m'$, the probability of $m' \in (l/2,l]$ is at most $1/2$. Hence the probability of any node $v$ having a min $m' \in (l/2,l]$ after 3 more rounds is 1/4.

Now since no set, $\set_k(m)$, ever get splits in subsequent rounds, the expected number of cores is 3l/4 after 3 more rounds. Hence in three rounds of \Hcm, the expected number of cores reduces  from $l$ to $3l/4$, and therefore it will terminate in expected $3\log{n}$ time.

The communication complexity is  $2(\nnodes + \nedges)$ per round in the worst-case since the total size of clusters is $\sum_{v} \core{v} = 2 (\nnodes)$.
\end{proof}

\section{Clustering Proofs}
\label{sec:proofs:clustering}
\subsection{Proof of Theorem~\ref{thm:clustering:correctness}}

We first restate Theorem~\ref{thm:clustering:correctness} below.

\begin{theorem}[\ref{thm:clustering:correctness}] The distributed Algorithm~\ref{algo:distributed:clustering} simulates the centralized Algorithm~\ref{algo:central}, i.e., it outputs the same clustering as Algorithm~\ref{algo:central}. 
\end{theorem}
\begin{proof} Let $\C_{central}$ be the clustering output by  Algorithm~\ref{algo:central}. Let $\C_{distributed}$ be the clustering output by Algorithm~\ref{algo:distributed:clustering}.  We show the result in two parts.

First, for any cluster $C_{distributed} \in \C_{distributed}$, there exists a cluster $C_{central} \in \C_{central}$ such that $C_{distributed} \subseteq C_{central}$.  Since Algorithm~\ref{algo:distributed:clustering} uses the splitting algorithm~\ref{algo:split}, it outputs only cores having cluster splits $C_l,C_r$ for which $\Stop_{local}(C_l)$ and $\Stop_{local}(C_r)$ equal false. Thus we can invoke Lemma~\ref{lemma:cores} on $C_{distributed}$ to prove that $C_{distributed}$ is valid, and the existence of $C_{central}$ such that $C_{distributed} \subseteq C_{central}$. 
 
Having shown that $C_{distributed} \subseteq C_{central}$, we now show that, in fact, $C_{distributed} = C_{central}$. Assume the contrary, i.e. $C_{distributed} \subset C_{central}$. Since $C_{distributed}$ is valid, even the centralized algorithm constructed $C_{distributed}$ some time during its execution, and then merged it with some other cluster, say $C'_{central}$. 

Since $C_{distributed}$ is in the output of Algorithm~\ref{algo:distributed:clustering}, then three cases are pospsible: (i) $C_{distributed}$ forms a connected component by itself, disconnected from the rest of the graph, or (ii) $\Stop_{local}(C_{distributed})$ is true, and the algorithm stops because of the stopping condition,  or (iii) $\Stop_{local}(C_{distributed})$ is false, but it merges with some cluster $C'_{distributed}$ for which $\Stop_{local}(C'_{distributed})$ is true. In the first two cases, even the centralized algorithm can not merge $C_{distributed}$ with any other cluster, contradicting that $C_{distributed} \subset C_{central}$. 

For case (iii), we show below that in fact $C'_{distributed} \subseteq C'_{central}$. Since the central algorithm merges $C'_{central}$ with $C_{central}$, $\Stop_{local}(C'_{central})$ has to be false. Since $\Stop_{local}$ is monotonic, and $C'_{distributed} \subseteq C'_{central}$, $\Stop_{local}(C'_{distributed})$ has to be false as well, contradicting the assumption made in case (iii). Thus we proved all three cases are impossible, contradicting our assumption of $C_{distributed} \subset C_{central}$. Hence $C_{distributed} = C_{central}$. 

Now we show that  $C'_{distributed} \subseteq C'_{central}$. Both $C'_{distributed}$ and $C'_{central}$ have to be closest to $C_{distributed}$, i.e. in $\nbrs(C_{distributed})$, in order to get merged with it in either the central or distributed algorithms. Denote $v$ to be the node, such that the singleton cluster $\{v\}$ is in $\nbrs(C_{distributed})$. Hence, by the property of \slc, both $C'_{distributed}$ and $C'_{central}$ must contain the node $v$. Since $C'_{distributed}$ is valid, there must be a cluster in the central algorithm's output containing it. Finally, since clusters in the output have to be disjoint, the cluster containing $C'_{distributed}$ has to be $C'_{central}$, and thus $C'_{distributed} \subseteq C'_{central}$.
\end{proof}

}
\end{document}